\documentclass[10pt]{article}
\usepackage[margin=0.7in]{geometry}
\usepackage{amsmath,amssymb,amsthm,graphicx,booktabs,array,tabularx,float,xcolor,listings}
\newtheorem{proposition}{Proposition}
\usepackage[T1]{fontenc}
\usepackage{lmodern}
\usepackage{quantikz}
\definecolor{NicePurple}{rgb}{0.3,0.1,1}
\usepackage[colorlinks=true,citecolor=NicePurple,linkcolor=NicePurple,urlcolor=NicePurple]{hyperref}
\definecolor{CiteRed}{HTML}{B7204D}
\definecolor{CiteOrange}{HTML}{AF4E16}
\definecolor{CiteGold}{HTML}{966500}
\definecolor{CiteGreen}{HTML}{007C55}
\definecolor{CiteTeal}{HTML}{007689}
\definecolor{CiteBlue}{HTML}{345BD0}
\colorlet{CiteViolet}{NicePurple}

\makeatletter
\providecommand{\RainbowCiteCount}{1}
\newcommand{\SetRainbowCiteCount}[1]{%
  \gdef\RainbowCiteCount{#1}%
  \if@filesw
    \immediate\write\@auxout{\string\gdef\string\RainbowCiteCount{#1}}%
  \fi
}
\newcommand{\RainbowSetColor}[1]{%
  \ifnum\RainbowCiteCount>1
    \@tempcnta=\numexpr600*(#1-1)/(\RainbowCiteCount-1)\relax
  \else
    \@tempcnta=600
  \fi
  \ifnum\@tempcnta<100
    \@tempcnta=\numexpr100-\@tempcnta\relax
    \color{CiteRed!\the\@tempcnta!CiteOrange}%
  \else\ifnum\@tempcnta<200
    \@tempcnta=\numexpr200-\@tempcnta\relax
    \color{CiteOrange!\the\@tempcnta!CiteGold}%
  \else\ifnum\@tempcnta<300
    \@tempcnta=\numexpr300-\@tempcnta\relax
    \color{CiteGold!\the\@tempcnta!CiteGreen}%
  \else\ifnum\@tempcnta<400
    \@tempcnta=\numexpr400-\@tempcnta\relax
    \color{CiteGreen!\the\@tempcnta!CiteTeal}%
  \else\ifnum\@tempcnta<500
    \@tempcnta=\numexpr500-\@tempcnta\relax
    \color{CiteTeal!\the\@tempcnta!CiteBlue}%
  \else
    \@tempcnta=\numexpr600-\@tempcnta\relax
    \color{CiteBlue!\the\@tempcnta!CiteViolet}%
  \fi\fi\fi\fi\fi
}

\makeatother

\newcommand{\CZTotalShots}{1,505,124,640}
\newcommand{\CZMaxBasisShotImbalancePercent}{5.12}

\newcommand{\CZTopologyText}{planar and torus}

\newcommand{\PlanarCrumbleLink}{\href{https://algassert.com/crumble\#circuit=Q(1.03033,\_6.33363)\_0;Q(1.56066,\_6.86396)\_1;Q(2.09099,\_7.39429)\_2;Q(1.03033,\_4.21231)\_3;Q(1.56066,\_4.74264)\_4;Q(2.09099,\_5.27297)\_5;Q(2.62132,\_5.8033)\_6;Q(3.15165,\_6.33363)\_7;Q(1.03033,\_2.09099)\_8;Q(1.56066,\_2.62132)\_9;Q(2.09099,\_3.15165)\_10;Q(2.62132,\_3.68198)\_11;Q(3.15165,\_4.21231)\_12;Q(3.68198,\_4.74264)\_13;Q(4.21231,\_5.27297)\_14;Q(2.09099,\_1.03033)\_15;Q(2.62132,\_1.56066)\_16;Q(3.15165,\_2.09099)\_17;Q(3.68198,\_2.62132)\_18;Q(4.21231,\_3.15165)\_19;Q(4.74264,\_3.68198)\_20;Q(5.27297,\_4.21231)\_21;Q(4.21231,\_1.03033)\_22;Q(4.74264,\_1.56066)\_23;Q(5.27297,\_2.09099)\_24;Q(5.8033,\_2.62132)\_25;Q(6.33363,\_3.15165)\_26;Q(6.33363,\_1.03033)\_27;Q(6.86396,\_1.56066)\_28;Q(7.39429,\_2.09099)\_29;Q(8.45495,\_1.03033)\_30;Q(0.5,\_4.74264)\_31;Q(1.03033,\_5.27297)\_32;Q(1.56066,\_5.8033)\_33;Q(2.09099,\_6.33363)\_34;Q(2.62132,\_6.86396)\_35;Q(3.15165,\_7.39429)\_36;Q(1.56066,\_3.68198)\_37;Q(2.09099,\_4.21231)\_38;Q(2.62132,\_4.74264)\_39;Q(3.15165,\_5.27297)\_40;Q(3.68198,\_5.8033)\_41;Q(1.56066,\_1.56066)\_42;Q(2.09099,\_2.09099)\_43;Q(2.62132,\_2.62132)\_44;Q(3.15165,\_3.15165)\_45;Q(3.68198,\_3.68198)\_46;Q(4.21231,\_4.21231)\_47;Q(4.74264,\_4.74264)\_48;Q(5.27297,\_5.27297)\_49;Q(2.62132,\_0.5)\_50;Q(3.15165,\_1.03033)\_51;Q(3.68198,\_1.56066)\_52;Q(4.21231,\_2.09099)\_53;Q(4.74264,\_2.62132)\_54;Q(5.27297,\_3.15165)\_55;Q(5.8033,\_3.68198)\_56;Q(5.8033,\_1.56066)\_57;Q(6.33363,\_2.09099)\_58;Q(6.86396,\_2.62132)\_59;Q(7.39429,\_3.15165)\_60;Q(6.86396,\_0.5)\_61;Q(7.39429,\_1.03033)\_62;Q(7.92462,\_1.56066)\_63;POLYGON(0,0,1,0.5)\_0\_5\_4\_3;POLYGON(1,0,0,0.5)\_3\_4\_11\_10;POLYGON(0,0,1,0.5)\_9\_10\_11\_12\_19\_18\_17\_16;POLYGON(1,0,0,0.5)\_8\_9\_16\_15;POLYGON(0,1,0,0.5)\_15\_16\_17\_22;POLYGON(1,0,0,0.5)\_17\_18\_23\_22;POLYGON(0,0,1,0.5)\_1\_2\_7\_6;POLYGON(1,0,0,0.5)\_0\_1\_6\_5;POLYGON(0,1,0,0.5)\_4\_5\_6\_7\_14\_13\_12\_11;POLYGON(0,0,1,0.5)\_13\_14\_21\_20;POLYGON(1,0,0,0.5)\_12\_13\_20\_19;POLYGON(0,1,0,0.5)\_18\_19\_20\_21\_26\_25\_24\_23;POLYGON(0,0,1,0.5)\_25\_26\_29\_28;POLYGON(1,0,0,0.5)\_24\_25\_28\_27;POLYGON(0,1,0,0.5)\_27\_28\_29\_30;TICK;RX\_31;R\_32;RX\_33;R\_34;RX\_35;R\_36\_37;RX\_38\_39;R\_40\_41\_42;RX\_43\_44;R\_45;RX\_46;R\_47;RX\_48;R\_49;RX\_50;R\_51;RX\_52;R\_53;RX\_54;R\_55\_56\_57;RX\_58\_59;R\_60\_61;RX\_62;R\_63;MARKX(0)\_39;MARKZ(1)\_45;TICK;CX\_31\_32\_33\_34\_38\_37\_39\_40\_43\_42\_44\_45\_46\_47\_48\_49\_50\_51\_52\_53\_54\_55\_58\_57\_62\_61;TICK;CX\_31\_3\_32\_33\_35\_34\_39\_38\_40\_41\_42\_8\_44\_17\_45\_46\_48\_47\_51\_52\_54\_53\_55\_56\_57\_27\_59\_58;TICK;CX\_32\_4\_33\_0\_35\_36\_39\_12\_40\_6\_41\_14\_42\_15\_44\_43\_45\_18\_46\_19\_47\_13\_51\_16\_52\_17\_53\_23\_55\_25\_56\_21\_58\_28\_59\_29\_61\_27;TICK;CX\_33\_5\_34\_6\_35\_2\_38\_4\_40\_13\_41\_7\_43\_9\_44\_10\_45\_11\_46\_12\_48\_14\_52\_22\_53\_18\_54\_24\_55\_20\_56\_26\_58\_25\_59\_60\_62\_61;TICK;CX\_32\_33\_34\_1\_35\_7\_37\_10\_38\_11\_39\_5\_40\_41\_43\_44\_45\_46\_47\_20\_48\_21\_51\_52\_54\_53\_55\_56\_59\_58\_62\_63;TICK;CX\_4\_32\_35\_34\_37\_3\_39\_38\_41\_40\_43\_16\_44\_45\_48\_47\_52\_51\_54\_19\_56\_55\_59\_26\_62\_28\_63\_29;TICK;CX\_3\_31\_7\_41\_9\_43\_13\_40\_17\_52\_21\_56\_25\_55\_26\_59\_34\_35\_38\_39\_45\_44\_47\_48\_53\_54\_57\_24;TICK;CX\_4\_38\_5\_39\_6\_34\_7\_35\_13\_47\_18\_53\_19\_54\_21\_48\_22\_52\_24\_57\_33\_32\_41\_40\_44\_43\_46\_45\_50\_15\_56\_55\_58\_59\_61\_62;TICK;CX\_0\_33\_1\_34\_2\_35\_7\_41\_10\_44\_11\_45\_12\_46\_15\_42\_16\_43\_20\_47\_21\_56\_23\_53\_24\_54\_25\_58\_27\_61\_29\_59\_31\_32\_52\_51\_63\_30;TICK;CX\_3\_37\_5\_33\_6\_40\_8\_42\_11\_38\_12\_39\_14\_41\_15\_50\_16\_51\_18\_45\_19\_46\_20\_55\_26\_56\_27\_57\_28\_58\_29\_63\_34\_35\_43\_44\_47\_48\_53\_54\_61\_62;TICK;CX\_10\_37\_14\_48\_17\_44\_28\_62\_30\_63\_33\_32\_38\_39\_41\_40\_43\_42\_46\_45\_50\_51\_52\_53\_56\_55\_58\_59;TICK;CX\_14\_41\_26\_56\_27\_61\_33\_34\_35\_36\_38\_37\_39\_40\_44\_45\_46\_47\_48\_49\_54\_55\_58\_57\_59\_60\_62\_63;TICK;MX\_31;M\_32;MX\_33;M\_34;MX\_35;M\_36\_37;MX\_38\_39;M\_40\_41\_42;MX\_43\_44;M\_45;MX\_46;M\_47;MX\_48;M\_49;MX\_50;M\_51;MX\_52;M\_53;MX\_54;M\_55\_56\_57;MX\_58\_59;M\_60\_61;MX\_62;M\_63;MARKX(0)\_33\_39\_46;MARKZ(1)\_45;TICK;RX\_31;R\_32;RX\_33;R\_34;RX\_35;R\_36\_37;RX\_38\_39;R\_40\_41\_42;RX\_43\_44;R\_45;RX\_46;R\_47;RX\_48;R\_49;RX\_50;R\_51;RX\_52;R\_53;RX\_54;R\_55\_56\_57;RX\_58\_59;R\_60\_61;RX\_62;R\_63;MARKX(0)\_33\_39\_46;MARKZ(1)\_45;TICK;CX\_14\_41\_26\_56\_27\_61\_33\_34\_35\_36\_38\_37\_39\_40\_44\_45\_46\_47\_48\_49\_54\_55\_58\_57\_59\_60\_62\_63;TICK;CX\_10\_37\_14\_48\_17\_44\_28\_62\_30\_63\_33\_32\_38\_39\_41\_40\_43\_42\_46\_45\_50\_51\_52\_53\_56\_55\_58\_59;TICK;CX\_3\_37\_5\_33\_6\_40\_8\_42\_11\_38\_12\_39\_14\_41\_15\_50\_16\_51\_18\_45\_19\_46\_20\_55\_26\_56\_27\_57\_28\_58\_29\_63\_34\_35\_43\_44\_47\_48\_53\_54\_61\_62;TICK;CX\_0\_33\_1\_34\_2\_35\_7\_41\_10\_44\_11\_45\_12\_46\_15\_42\_16\_43\_20\_47\_21\_56\_23\_53\_24\_54\_25\_58\_27\_61\_29\_59\_31\_32\_52\_51\_63\_30;TICK;CX\_4\_38\_5\_39\_6\_34\_7\_35\_13\_47\_18\_53\_19\_54\_21\_48\_22\_52\_24\_57\_33\_32\_41\_40\_44\_43\_46\_45\_50\_15\_56\_55\_58\_59\_61\_62;TICK;CX\_3\_31\_7\_41\_9\_43\_13\_40\_17\_52\_21\_56\_25\_55\_26\_59\_34\_35\_38\_39\_45\_44\_47\_48\_53\_54\_57\_24;TICK;CX\_4\_32\_35\_34\_37\_3\_39\_38\_41\_40\_43\_16\_44\_45\_48\_47\_52\_51\_54\_19\_56\_55\_59\_26\_62\_28\_63\_29;TICK;CX\_32\_33\_34\_1\_35\_7\_37\_10\_38\_11\_39\_5\_40\_41\_43\_44\_45\_46\_47\_20\_48\_21\_51\_52\_54\_53\_55\_56\_59\_58\_62\_63;TICK;CX\_33\_5\_34\_6\_35\_2\_38\_4\_40\_13\_41\_7\_43\_9\_44\_10\_45\_11\_46\_12\_48\_14\_52\_22\_53\_18\_54\_24\_55\_20\_56\_26\_58\_25\_59\_60\_62\_61;TICK;CX\_32\_4\_33\_0\_35\_36\_39\_12\_40\_6\_41\_14\_42\_15\_44\_43\_45\_18\_46\_19\_47\_13\_51\_16\_52\_17\_53\_23\_55\_25\_56\_21\_58\_28\_59\_29\_61\_27;TICK;CX\_31\_3\_32\_33\_35\_34\_39\_38\_40\_41\_42\_8\_44\_17\_45\_46\_48\_47\_51\_52\_54\_53\_55\_56\_57\_27\_59\_58;TICK;CX\_31\_32\_33\_34\_38\_37\_39\_40\_43\_42\_44\_45\_46\_47\_48\_49\_50\_51\_52\_53\_54\_55\_58\_57\_62\_61;TICK;MX\_31;M\_32;MX\_33;M\_34;MX\_35;M\_36\_37;MX\_38\_39;M\_40\_41\_42;MX\_43\_44;M\_45;MX\_46;M\_47;MX\_48;M\_49;MX\_50;M\_51;MX\_52;M\_53;MX\_54;M\_55\_56\_57;MX\_58\_59;M\_60\_61;MX\_62;M\_63;MARKX(0)\_39;MARKZ(1)\_45}{\textcolor{NicePurple}{Open the triangular superdense 4.8.8 color code circuit in Crumble.}}}

\newcommand{\TorusCrumbleLink}{\href{https://algassert.com/crumble\#circuit=Q(4.35,\_8.75)\_0;Q(4.35,\_0.5)\_1;Q(4.35,\_1.05)\_2;Q(4.35,\_1.6)\_3;Q(1.05,\_4.35)\_4;Q(1.05,\_4.9)\_5;Q(1.05,\_5.45)\_6;Q(5.45,\_1.6)\_7;Q(4.35,\_2.15)\_8;Q(4.35,\_2.7)\_9;Q(4.35,\_3.25)\_10;Q(4.35,\_3.8)\_11;Q(5.45,\_2.15)\_12;Q(5.45,\_2.7)\_13;Q(5.45,\_3.25)\_14;Q(5.45,\_3.8)\_15;Q(6.55,\_6.55)\_16;Q(2.15,\_2.7)\_17;Q(2.15,\_3.25)\_18;Q(2.15,\_3.8)\_19;Q(3.25,\_2.15)\_20;Q(3.25,\_2.7)\_21;Q(3.25,\_3.25)\_22;Q(3.25,\_3.8)\_23;Q(2.15,\_4.35)\_24;Q(2.15,\_4.9)\_25;Q(2.15,\_5.45)\_26;Q(2.15,\_6)\_27;Q(3.25,\_4.35)\_28;Q(3.25,\_4.9)\_29;Q(3.25,\_5.45)\_30;Q(3.25,\_6)\_31;Q(2.15,\_6.55)\_32;Q(6.55,\_2.7)\_33;Q(6.55,\_3.25)\_34;Q(6.55,\_3.8)\_35;Q(3.25,\_6.55)\_36;Q(3.25,\_7.1)\_37;Q(3.25,\_7.65)\_38;Q(7.65,\_3.8)\_39;Q(6.55,\_4.35)\_40;Q(6.55,\_4.9)\_41;Q(6.55,\_5.45)\_42;Q(6.55,\_6)\_43;Q(7.65,\_4.35)\_44;Q(7.65,\_4.9)\_45;Q(7.65,\_5.45)\_46;Q(3.25,\_1.6)\_47;Q(4.35,\_4.35)\_48;Q(4.35,\_4.9)\_49;Q(4.35,\_5.45)\_50;Q(4.35,\_6)\_51;Q(5.45,\_4.35)\_52;Q(5.45,\_4.9)\_53;Q(5.45,\_5.45)\_54;Q(5.45,\_6)\_55;Q(4.35,\_6.55)\_56;Q(4.35,\_7.1)\_57;Q(4.35,\_7.65)\_58;Q(4.35,\_8.2)\_59;Q(5.45,\_6.55)\_60;Q(5.45,\_7.1)\_61;Q(5.45,\_7.65)\_62;Q(1.05,\_3.8)\_63;Q(3.8,\_7.65)\_64;Q(3.8,\_8.2)\_65;Q(6,\_6.55)\_66;Q(6,\_7.1)\_67;Q(2.7,\_6.55)\_68;Q(2.7,\_7.1)\_69;Q(4.9,\_5.45)\_70;Q(4.9,\_6)\_71;Q(3.8,\_5.45)\_72;Q(3.8,\_6)\_73;Q(3.8,\_6.55)\_74;Q(3.8,\_7.1)\_75;Q(4.9,\_6.55)\_76;Q(4.9,\_7.1)\_77;Q(4.9,\_7.65)\_78;Q(4.9,\_8.2)\_79;Q(6,\_5.45)\_80;Q(6,\_6)\_81;Q(8.2,\_4.35)\_82;Q(8.2,\_4.9)\_83;Q(4.9,\_4.35)\_84;Q(4.9,\_4.9)\_85;Q(7.1,\_3.25)\_86;Q(7.1,\_3.8)\_87;Q(6,\_3.25)\_88;Q(6,\_3.8)\_89;Q(6,\_4.35)\_90;Q(6,\_4.9)\_91;Q(7.1,\_4.35)\_92;Q(7.1,\_4.9)\_93;Q(7.1,\_5.45)\_94;Q(7.1,\_6)\_95;Q(1.6,\_5.45)\_96;Q(1.6,\_6)\_97;Q(3.8,\_4.35)\_98;Q(3.8,\_4.9)\_99;Q(0.5,\_4.35)\_100;Q(0.5,\_4.9)\_101;Q(2.7,\_3.25)\_102;Q(2.7,\_3.8)\_103;Q(1.6,\_3.25)\_104;Q(1.6,\_3.8)\_105;Q(1.6,\_4.35)\_106;Q(1.6,\_4.9)\_107;Q(2.7,\_4.35)\_108;Q(2.7,\_4.9)\_109;Q(2.7,\_5.45)\_110;Q(2.7,\_6)\_111;Q(3.8,\_3.25)\_112;Q(3.8,\_3.8)\_113;Q(6,\_2.15)\_114;Q(6,\_2.7)\_115;Q(2.7,\_2.15)\_116;Q(2.7,\_2.7)\_117;Q(4.9,\_1.05)\_118;Q(4.9,\_1.6)\_119;Q(3.8,\_1.05)\_120;Q(3.8,\_1.6)\_121;Q(3.8,\_2.15)\_122;Q(3.8,\_2.7)\_123;Q(4.9,\_2.15)\_124;Q(4.9,\_2.7)\_125;Q(4.9,\_3.25)\_126;Q(4.9,\_3.8)\_127;POLYGON(1,0,0,0.5)\_5\_6\_26\_25;POLYGON(1,0,0,0.5)\_3\_8\_12\_7;POLYGON(0,1,0,0.5)\_8\_9\_10\_11\_15\_14\_13\_12;POLYGON(0,0,1,0.5)\_22\_23\_28\_29\_49\_48\_11\_10;POLYGON(1,0,0,0.5)\_21\_22\_10\_9;POLYGON(1,0,0,0.5)\_19\_24\_28\_23;POLYGON(0,1,0,0.5)\_24\_25\_26\_27\_31\_30\_29\_28;POLYGON(1,0,0,0.5)\_37\_38\_58\_57;POLYGON(1,0,0,0.5)\_35\_40\_44\_39;POLYGON(0,0,1,0.5)\_14\_15\_52\_53\_41\_40\_35\_34;POLYGON(1,0,0,0.5)\_13\_14\_34\_33;POLYGON(1,0,0,0.5)\_27\_32\_36\_31;POLYGON(0,1,0,0.5)\_48\_49\_50\_51\_55\_54\_53\_52;POLYGON(1,0,0,0.5)\_53\_54\_42\_41;POLYGON(1,0,0,0.5)\_51\_56\_60\_55;POLYGON(0,0,1,0.5)\_30\_31\_36\_37\_57\_56\_51\_50;POLYGON(1,0,0,0.5)\_29\_30\_50\_49;POLYGON(1,0,0,0.5)\_11\_48\_52\_15;TICK;R\_64;RX\_65;R\_66;RX\_67;R\_68;RX\_69;R\_70;RX\_71;R\_72;RX\_73;R\_74;RX\_75;R\_76\_77;RX\_78\_79;R\_80;RX\_81;R\_82;RX\_83;R\_84;RX\_85;R\_86;RX\_87;R\_88;RX\_89;R\_90;RX\_91;R\_92\_93;RX\_94\_95;R\_96;RX\_97;R\_98;RX\_99;R\_100;RX\_101;R\_102;RX\_103;R\_104;RX\_105;R\_106;RX\_107;R\_108\_109;RX\_110\_111;R\_112;RX\_113;R\_114;RX\_115;R\_116;RX\_117;R\_118;RX\_119;R\_120;RX\_121;R\_122;RX\_123;R\_124\_125;RX\_126\_127;MARKX(0)\_85\_91;MARKZ(1)\_74\_76;TICK;CX\_75\_64\_81\_66\_69\_86\_71\_76\_99\_72\_73\_74\_78\_77\_79\_100\_91\_80\_65\_82\_85\_70\_87\_92\_115\_88\_89\_90\_94\_93\_95\_116\_107\_96\_113\_98\_101\_118\_103\_108\_67\_104\_105\_106\_110\_109\_111\_68\_123\_112\_97\_114\_117\_102\_119\_124\_83\_120\_121\_122\_126\_125\_127\_84;TICK;CX\_64\_65\_67\_66\_69\_68\_70\_71\_72\_73\_75\_74\_77\_76\_78\_79\_80\_81\_83\_82\_85\_84\_86\_87\_88\_89\_91\_90\_93\_92\_94\_95\_96\_97\_99\_98\_101\_100\_102\_103\_104\_105\_107\_106\_109\_108\_110\_111\_112\_113\_115\_114\_117\_116\_118\_119\_120\_121\_123\_122\_125\_124\_126\_127;TICK;CX\_59\_65\_16\_66\_68\_36\_71\_55\_51\_73\_56\_74\_76\_60\_79\_63\_43\_81\_0\_82\_84\_52\_87\_39\_35\_89\_40\_90\_92\_44\_95\_47\_27\_97\_48\_98\_100\_4\_103\_23\_19\_105\_24\_106\_108\_28\_111\_31\_11\_113\_32\_114\_116\_20\_119\_7\_3\_121\_8\_122\_124\_12\_127\_15\_58\_64\_17\_67\_69\_37\_70\_54\_50\_72\_57\_75\_77\_61\_78\_62\_42\_80\_1\_83\_85\_53\_86\_38\_34\_88\_41\_91\_93\_45\_94\_46\_26\_96\_49\_99\_101\_5\_102\_22\_18\_104\_25\_107\_109\_29\_110\_30\_10\_112\_33\_115\_117\_21\_118\_6\_2\_120\_9\_123\_125\_13\_126\_14;TICK;CX\_38\_64\_61\_67\_69\_33\_70\_50\_30\_72\_37\_75\_77\_57\_78\_58\_54\_80\_45\_83\_85\_49\_86\_34\_14\_88\_53\_91\_93\_41\_94\_42\_6\_96\_29\_99\_101\_1\_102\_18\_62\_104\_5\_107\_109\_25\_110\_26\_22\_112\_13\_115\_117\_17\_118\_2\_46\_120\_21\_123\_125\_9\_126\_10\_39\_65\_60\_66\_68\_32\_71\_51\_31\_73\_36\_74\_76\_56\_79\_59\_55\_81\_44\_82\_84\_48\_87\_35\_15\_89\_52\_90\_92\_40\_95\_43\_7\_97\_28\_98\_100\_0\_103\_19\_63\_105\_4\_106\_108\_24\_111\_27\_23\_113\_12\_114\_116\_16\_119\_3\_47\_121\_20\_122\_124\_8\_127\_11;TICK;CX\_66\_67\_80\_81\_68\_69\_86\_87\_74\_75\_72\_73\_79\_78\_77\_76\_82\_83\_64\_65\_84\_85\_70\_71\_90\_91\_88\_89\_95\_94\_93\_92\_98\_99\_112\_113\_100\_101\_118\_119\_106\_107\_104\_105\_111\_110\_109\_108\_114\_115\_96\_97\_116\_117\_102\_103\_122\_123\_120\_121\_127\_126\_125\_124;TICK;CX\_81\_66\_69\_86\_73\_74\_78\_77\_65\_82\_85\_70\_89\_90\_94\_93\_113\_98\_101\_118\_105\_106\_110\_109\_97\_114\_117\_102\_121\_122\_126\_125;TICK;CX\_66\_81\_86\_69\_74\_73\_77\_78\_82\_65\_70\_85\_90\_89\_93\_94\_98\_113\_118\_101\_106\_105\_109\_110\_114\_97\_102\_117\_122\_121\_125\_126;TICK;CX\_67\_66\_81\_80\_69\_68\_87\_86\_75\_74\_73\_72\_78\_79\_76\_77\_83\_82\_65\_64\_85\_84\_71\_70\_91\_90\_89\_88\_94\_95\_92\_93\_99\_98\_113\_112\_101\_100\_119\_118\_107\_106\_105\_104\_110\_111\_108\_109\_115\_114\_97\_96\_117\_116\_103\_102\_123\_122\_121\_120\_126\_127\_124\_125;TICK;CX\_83\_1\_120\_2\_5\_101\_6\_118\_123\_9\_112\_10\_13\_125\_14\_126\_67\_17\_104\_18\_21\_117\_22\_102\_107\_25\_96\_26\_29\_109\_30\_110\_115\_33\_88\_34\_37\_69\_38\_86\_91\_41\_80\_42\_45\_93\_46\_94\_99\_49\_72\_50\_53\_85\_54\_70\_75\_57\_64\_58\_61\_77\_62\_78\_82\_0\_121\_3\_4\_100\_7\_119\_122\_8\_113\_11\_12\_124\_15\_127\_66\_16\_105\_19\_20\_116\_23\_103\_106\_24\_97\_27\_28\_108\_31\_111\_114\_32\_89\_35\_36\_68\_39\_87\_90\_40\_81\_43\_44\_92\_47\_95\_98\_48\_73\_51\_52\_84\_55\_71\_74\_56\_65\_59\_60\_76\_63\_79;TICK;CX\_0\_100\_3\_119\_106\_4\_97\_7\_8\_124\_11\_127\_114\_12\_89\_15\_16\_116\_19\_103\_122\_20\_113\_23\_24\_108\_27\_111\_98\_28\_73\_31\_32\_68\_35\_87\_74\_36\_65\_39\_40\_92\_43\_95\_82\_44\_121\_47\_48\_84\_51\_71\_90\_52\_81\_55\_56\_76\_59\_79\_66\_60\_105\_63\_1\_101\_2\_118\_107\_5\_96\_6\_9\_125\_10\_126\_115\_13\_88\_14\_17\_117\_18\_102\_123\_21\_112\_22\_25\_109\_26\_110\_99\_29\_72\_30\_33\_69\_34\_86\_75\_37\_64\_38\_41\_93\_42\_94\_83\_45\_120\_46\_49\_85\_50\_70\_91\_53\_80\_54\_57\_77\_58\_78\_67\_61\_104\_62;TICK;CX\_65\_64\_66\_67\_68\_69\_71\_70\_73\_72\_74\_75\_76\_77\_79\_78\_81\_80\_82\_83\_84\_85\_87\_86\_89\_88\_90\_91\_92\_93\_95\_94\_97\_96\_98\_99\_100\_101\_103\_102\_105\_104\_106\_107\_108\_109\_111\_110\_113\_112\_114\_115\_116\_117\_119\_118\_121\_120\_122\_123\_124\_125\_127\_126;TICK;CX\_75\_64\_81\_66\_69\_86\_71\_76\_99\_72\_73\_74\_78\_77\_79\_100\_91\_80\_65\_82\_85\_70\_87\_92\_115\_88\_89\_90\_94\_93\_95\_116\_107\_96\_113\_98\_101\_118\_103\_108\_67\_104\_105\_106\_110\_109\_111\_68\_123\_112\_97\_114\_117\_102\_119\_124\_83\_120\_121\_122\_126\_125\_127\_84;TICK;M\_64;MX\_65;M\_66;MX\_67;M\_68;MX\_69;M\_70;MX\_71;M\_72;MX\_73;M\_74;MX\_75;M\_76\_77;MX\_78\_79;M\_80;MX\_81;M\_82;MX\_83;M\_84;MX\_85;M\_86;MX\_87;M\_88;MX\_89;M\_90;MX\_91;M\_92\_93;MX\_94\_95;M\_96;MX\_97;M\_98;MX\_99;M\_100;MX\_101;M\_102;MX\_103;M\_104;MX\_105;M\_106;MX\_107;M\_108\_109;MX\_110\_111;M\_112;MX\_113;M\_114;MX\_115;M\_116;MX\_117;M\_118;MX\_119;M\_120;MX\_121;M\_122;MX\_123;M\_124\_125;MX\_126\_127;MARKX(0)\_85;MARKZ(1)\_74;TICK;R\_64;RX\_65;R\_66;RX\_67;R\_68;RX\_69;R\_70;RX\_71;R\_72;RX\_73;R\_74;RX\_75;R\_76\_77;RX\_78\_79;R\_80;RX\_81;R\_82;RX\_83;R\_84;RX\_85;R\_86;RX\_87;R\_88;RX\_89;R\_90;RX\_91;R\_92\_93;RX\_94\_95;R\_96;RX\_97;R\_98;RX\_99;R\_100;RX\_101;R\_102;RX\_103;R\_104;RX\_105;R\_106;RX\_107;R\_108\_109;RX\_110\_111;R\_112;RX\_113;R\_114;RX\_115;R\_116;RX\_117;R\_118;RX\_119;R\_120;RX\_121;R\_122;RX\_123;R\_124\_125;RX\_126\_127;MARKX(0)\_85\_99;MARKZ(1)\_68\_74;TICK;CX\_75\_64\_81\_66\_69\_86\_71\_76\_99\_72\_73\_74\_78\_77\_79\_100\_91\_80\_65\_82\_85\_70\_87\_92\_115\_88\_89\_90\_94\_93\_95\_116\_107\_96\_113\_98\_101\_118\_103\_108\_67\_104\_105\_106\_110\_109\_111\_68\_123\_112\_97\_114\_117\_102\_119\_124\_83\_120\_121\_122\_126\_125\_127\_84;TICK;CX\_64\_65\_67\_66\_69\_68\_70\_71\_72\_73\_75\_74\_77\_76\_78\_79\_80\_81\_83\_82\_85\_84\_86\_87\_88\_89\_91\_90\_93\_92\_94\_95\_96\_97\_99\_98\_101\_100\_102\_103\_104\_105\_107\_106\_109\_108\_110\_111\_112\_113\_115\_114\_117\_116\_118\_119\_120\_121\_123\_122\_125\_124\_126\_127;TICK;CX\_59\_65\_16\_66\_68\_36\_71\_55\_51\_73\_56\_74\_76\_60\_79\_63\_43\_81\_0\_82\_84\_52\_87\_39\_35\_89\_40\_90\_92\_44\_95\_47\_27\_97\_48\_98\_100\_4\_103\_23\_19\_105\_24\_106\_108\_28\_111\_31\_11\_113\_32\_114\_116\_20\_119\_7\_3\_121\_8\_122\_124\_12\_127\_15\_58\_64\_17\_67\_69\_37\_70\_54\_50\_72\_57\_75\_77\_61\_78\_62\_42\_80\_1\_83\_85\_53\_86\_38\_34\_88\_41\_91\_93\_45\_94\_46\_26\_96\_49\_99\_101\_5\_102\_22\_18\_104\_25\_107\_109\_29\_110\_30\_10\_112\_33\_115\_117\_21\_118\_6\_2\_120\_9\_123\_125\_13\_126\_14;TICK;CX\_38\_64\_61\_67\_69\_33\_70\_50\_30\_72\_37\_75\_77\_57\_78\_58\_54\_80\_45\_83\_85\_49\_86\_34\_14\_88\_53\_91\_93\_41\_94\_42\_6\_96\_29\_99\_101\_1\_102\_18\_62\_104\_5\_107\_109\_25\_110\_26\_22\_112\_13\_115\_117\_17\_118\_2\_46\_120\_21\_123\_125\_9\_126\_10\_39\_65\_60\_66\_68\_32\_71\_51\_31\_73\_36\_74\_76\_56\_79\_59\_55\_81\_44\_82\_84\_48\_87\_35\_15\_89\_52\_90\_92\_40\_95\_43\_7\_97\_28\_98\_100\_0\_103\_19\_63\_105\_4\_106\_108\_24\_111\_27\_23\_113\_12\_114\_116\_16\_119\_3\_47\_121\_20\_122\_124\_8\_127\_11;TICK;CX\_66\_67\_80\_81\_68\_69\_86\_87\_74\_75\_72\_73\_79\_78\_77\_76\_82\_83\_64\_65\_84\_85\_70\_71\_90\_91\_88\_89\_95\_94\_93\_92\_98\_99\_112\_113\_100\_101\_118\_119\_106\_107\_104\_105\_111\_110\_109\_108\_114\_115\_96\_97\_116\_117\_102\_103\_122\_123\_120\_121\_127\_126\_125\_124;TICK;CX\_81\_66\_69\_86\_73\_74\_78\_77\_65\_82\_85\_70\_89\_90\_94\_93\_113\_98\_101\_118\_105\_106\_110\_109\_97\_114\_117\_102\_121\_122\_126\_125;TICK;CX\_66\_81\_86\_69\_74\_73\_77\_78\_82\_65\_70\_85\_90\_89\_93\_94\_98\_113\_118\_101\_106\_105\_109\_110\_114\_97\_102\_117\_122\_121\_125\_126;TICK;CX\_67\_66\_81\_80\_69\_68\_87\_86\_75\_74\_73\_72\_78\_79\_76\_77\_83\_82\_65\_64\_85\_84\_71\_70\_91\_90\_89\_88\_94\_95\_92\_93\_99\_98\_113\_112\_101\_100\_119\_118\_107\_106\_105\_104\_110\_111\_108\_109\_115\_114\_97\_96\_117\_116\_103\_102\_123\_122\_121\_120\_126\_127\_124\_125;TICK;CX\_83\_1\_120\_2\_5\_101\_6\_118\_123\_9\_112\_10\_13\_125\_14\_126\_67\_17\_104\_18\_21\_117\_22\_102\_107\_25\_96\_26\_29\_109\_30\_110\_115\_33\_88\_34\_37\_69\_38\_86\_91\_41\_80\_42\_45\_93\_46\_94\_99\_49\_72\_50\_53\_85\_54\_70\_75\_57\_64\_58\_61\_77\_62\_78\_82\_0\_121\_3\_4\_100\_7\_119\_122\_8\_113\_11\_12\_124\_15\_127\_66\_16\_105\_19\_20\_116\_23\_103\_106\_24\_97\_27\_28\_108\_31\_111\_114\_32\_89\_35\_36\_68\_39\_87\_90\_40\_81\_43\_44\_92\_47\_95\_98\_48\_73\_51\_52\_84\_55\_71\_74\_56\_65\_59\_60\_76\_63\_79;TICK;CX\_0\_100\_3\_119\_106\_4\_97\_7\_8\_124\_11\_127\_114\_12\_89\_15\_16\_116\_19\_103\_122\_20\_113\_23\_24\_108\_27\_111\_98\_28\_73\_31\_32\_68\_35\_87\_74\_36\_65\_39\_40\_92\_43\_95\_82\_44\_121\_47\_48\_84\_51\_71\_90\_52\_81\_55\_56\_76\_59\_79\_66\_60\_105\_63\_1\_101\_2\_118\_107\_5\_96\_6\_9\_125\_10\_126\_115\_13\_88\_14\_17\_117\_18\_102\_123\_21\_112\_22\_25\_109\_26\_110\_99\_29\_72\_30\_33\_69\_34\_86\_75\_37\_64\_38\_41\_93\_42\_94\_83\_45\_120\_46\_49\_85\_50\_70\_91\_53\_80\_54\_57\_77\_58\_78\_67\_61\_104\_62;TICK;CX\_65\_64\_66\_67\_68\_69\_71\_70\_73\_72\_74\_75\_76\_77\_79\_78\_81\_80\_82\_83\_84\_85\_87\_86\_89\_88\_90\_91\_92\_93\_95\_94\_97\_96\_98\_99\_100\_101\_103\_102\_105\_104\_106\_107\_108\_109\_111\_110\_113\_112\_114\_115\_116\_117\_119\_118\_121\_120\_122\_123\_124\_125\_127\_126;TICK;CX\_75\_64\_81\_66\_69\_86\_71\_76\_99\_72\_73\_74\_78\_77\_79\_100\_91\_80\_65\_82\_85\_70\_87\_92\_115\_88\_89\_90\_94\_93\_95\_116\_107\_96\_113\_98\_101\_118\_103\_108\_67\_104\_105\_106\_110\_109\_111\_68\_123\_112\_97\_114\_117\_102\_119\_124\_83\_120\_121\_122\_126\_125\_127\_84;TICK;M\_64;MX\_65;M\_66;MX\_67;M\_68;MX\_69;M\_70;MX\_71;M\_72;MX\_73;M\_74;MX\_75;M\_76\_77;MX\_78\_79;M\_80;MX\_81;M\_82;MX\_83;M\_84;MX\_85;M\_86;MX\_87;M\_88;MX\_89;M\_90;MX\_91;M\_92\_93;MX\_94\_95;M\_96;MX\_97;M\_98;MX\_99;M\_100;MX\_101;M\_102;MX\_103;M\_104;MX\_105;M\_106;MX\_107;M\_108\_109;MX\_110\_111;M\_112;MX\_113;M\_114;MX\_115;M\_116;MX\_117;M\_118;MX\_119;M\_120;MX\_121;M\_122;MX\_123;M\_124\_125;MX\_126\_127;MARKX(0)\_85;MARKZ(1)\_74}{\textcolor{NicePurple}{Open a $d=8$ torus superdense 4.8.8 color code circuit in Crumble.}}}

\usepackage[font=small,labelfont=bf]{caption}
\title{Denser Planar Color Codes}
\author{Noah Shutty\\[10pt]Google Quantum AI, Venice, CA, 90291}
\date{September 17, 2026}
\begin{document}
% Balance whitespace above the title and below Figure 1's caption.
\newgeometry{top=0.75in,bottom=0.7in,left=0.7in,right=0.7in}
\begingroup
\setlength{\parskip}{0pt}
\maketitle
\endgroup
\begin{abstract}
We consider the 4.8.8 color code in a brickwork layout on the square lattice
with nearest-neighbor gates.
We identify periodic superdense circuits with twelve CNOT layers that appear to preserve the full distance of these codes.
Including ancillas, the family uses
$q=(2d^2+5d-5)/2$ physical qubits for distance $d$: asymptotically $2/3$ as many
as the triangular 6.6.6 color code and $1/2$ as many as the rotated surface code
at the same code distance. Circuit-level simulations targeting
superconducting architectures demonstrate
favorable scaling of 4.8.8 color codes at plausible physical noise rates,
including both for planar codes and their toroidal cousins under periodic
boundary conditions.
\end{abstract}

{\small\textbf{Data availability:} The source, circuit, and data bundle is available at
\href{https://doi.org/10.5281/zenodo.22820614}{Zenodo (10.5281/zenodo.22820614)}.\par}

% Preserve Figure 1's position after shortening the abstract.
\vspace{22pt}

\begin{figure}[H]
\centering
\includegraphics[width=0.85\linewidth]{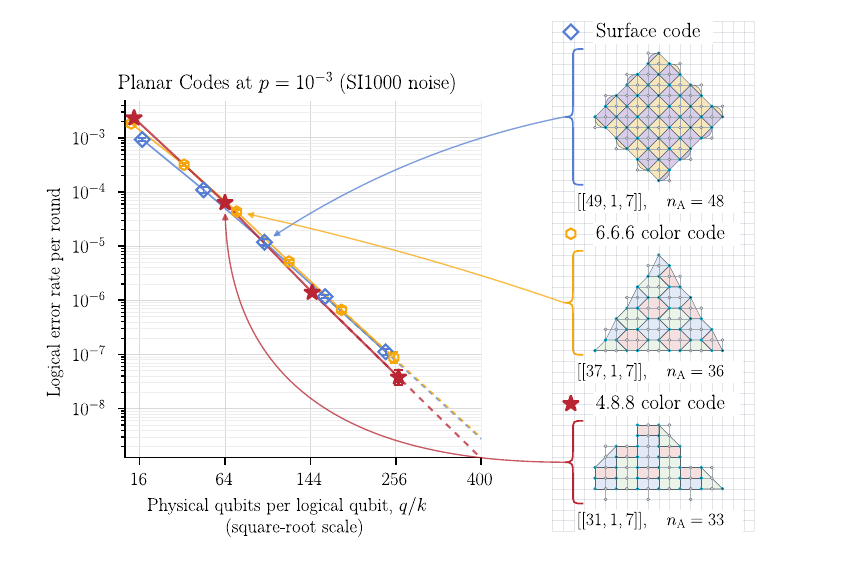}
\caption{Denser planar quantum memories. Left: pooled $X/Z$ memory
performance at $p=10^{-3}$ under SI1000 noise with native CZ gates.
Error bars represent 95\% confidence intervals.
Dashed lines extrapolate $\log\ell$ linearly in $\sqrt{q/k}$ through the two
largest sizes with observed errors. Decoding uses
Tesseract~\cite{tesseract} with parameters described
in Appendix~\ref{app:simulation}.}
\label{fig:feature}
\end{figure}

\restoregeometry
\section{Introduction}
Topological quantum error-correcting codes offer a route from noisy local
operations to robust quantum information storage. Beginning with the toric
code~\cite{toric}, they have played a central role in the history of fault-tolerant
quantum computing. Surface code architectures supplied early blueprints
for universal fault-tolerant computation with local operations, combining magic-state distillation with
defect braiding~\cite{topological_universal} and later lattice surgery~\cite{lattice_surgery}.
Even in the modern push towards higher-rate quantum low-density parity-check
(qLDPC) codes, topological codes have found use as building blocks for
modular architectures based on hyperbolic codes~\cite{higgott2026}.
Color codes~\cite{color} support transversal Clifford gates and efficient
magic-state preparation, including cultivation~\cite{cultivation}.
On the experimental side, color code memories and logical operations have been
explored on superconducting quantum processors~\cite{lacroix,rosenfeld}.
The local checks and high thresholds of topological codes remain attractive
even with reconfigurable hardware, with surface code experiments in neutral-atom
platforms demonstrating below-threshold error suppression~\cite{neutralatoms}.

For stabilizer codes with local checks in a two-dimensional Euclidean lattice,
the Bravyi--Poulin--Terhal (BPT) bound requires $n / (kd^2)=\Omega(1)$ for $n$ data qubits
encoding $k$ logical qubits at distance $d$~\cite{bpt}.
However, an interesting open question remains the constant-factor overhead $n / (kd^2)$ that can be achieved in practice.
It was observed in \cite{bmd07} that planar triangular 6.6.6 color codes offer
lower overhead than rotated surface codes, with $n/(kd^2)\to 3/4$ rather than $1$.
It was subsequently observed in \cite{landahl} that triangular 4.8.8
color codes offer lower overhead than triangular 6.6.6 color codes,
with $n/(kd^2)\to 1/2$ rather than $3/4$.
A major question is whether these savings in data qubits translate into lower
total qubit overhead at a target logical error rate once ancilla qubits and
circuit-level noise are included.
To provide a fair comparison between surface and color codes, the simulations
in~\cite{lacroix}
used a mixed-integer-program decoder to find the most likely error configuration
consistent with the observed detection events. These simulations showed that
6.6.6 color code memories can improve
on surface code overhead at realistic noise rates.

Here, we revisit the 4.8.8 color code and optimize its performance as a quantum memory.
Although these 4.8.8 color codes are often arranged in a square-octagon
layout~\cite{landahl,flagged_weights},
we use a brickwork embedding with alternating rails of data and ancilla qubits.
This layout was inspired by the brickwork ancilla-free code state shown
in~\cite{superdense} in the context of a middle-out circuit construction~\cite{midout}.
We also use Tesseract~\cite{tesseract}, which achieves accuracy comparable to
integer-program decoding at substantially lower runtime on benchmark circuits.
Finally, following the design principles of superdense syndrome
extraction~\cite{superdense}\footnote{Ref.~\cite{superdense} incorrectly gives
a distance of $d/2$ for its superdense 6.6.6 color code circuits.
These circuits preserve distance $d$, as clarified in~\cite{superdense_correction}.},
we construct a periodic circuit family that
preserves distance in all tested instances. It uses twelve CNOT layers,
compared with eight for the triangular 6.6.6 color code, while extracting all
$X$ and $Z$ stabilizers within a single reset/measure cycle.
Together, these improvements make the 4.8.8 color code a competitive quantum
memory on square-lattice hardware.

\section{Methods}
\label{sec:methods}
\subsection{Code construction}
\label{sec:code-construction}
Two-dimensional color codes place qubits on a three-colorable tiling, with
paired $X$ and $Z$ face checks~\cite{color}. Just as rotated and unrotated
surface codes use the same stabilizer tiling but different cuts, with the
rotated code more compact at fixed distance~\cite{lattice_surgery}, different
planar boundaries or periodic identifications change the overhead of color codes.
For example, Figure~\ref{fig:methods} shows the best $[[18,4,4]]$ and
alternate $[[24,4,4]]$ 6.6.6 torus quotients: the same $k$ and $d$ with
different numbers of data qubits.
For torus codes, we enumerate finite-index subgroups of the tiling's admissible
translations. Their fundamental regions are \emph{admissible parallelograms}:
identifying opposite edges matches qubits and checks consistently.

For planar codes, we select triangular color code patches or square surface
code patches, retaining interior checks and selected restrictions of
boundary-crossing checks to the remaining data qubits.
Each retained color-code face supports both an $X$ and a $Z$ check.
We select boundary fragments according to the boundary
type: a different face color is omitted on each side of a color code
triangle~\cite{boundaries}, while surface code squares have $X$-type boundary
checks on two opposite sides and $Z$-type checks on the other two. We verify
that the resulting stabilizers commute and encode one logical qubit.
Figure~\ref{fig:methods} illustrates these constructions.

\clearpage
% Preserve the page's RGB compositing group, while Figure 2 itself stays
% ungrouped for Apple PDF renderers (see FIGURE2_VECTOR_COMPATIBILITY.md).
\AddToHookNext{shipout/before}{\pdfpageattr{/Group << /Type /Group /S /Transparency /I true /CS /DeviceRGB >>}}
\AddToHookNext{shipout/after}{\pdfpageattr{}}
\begin{figure}[H]
\centering
\includegraphics[width=0.991\linewidth]{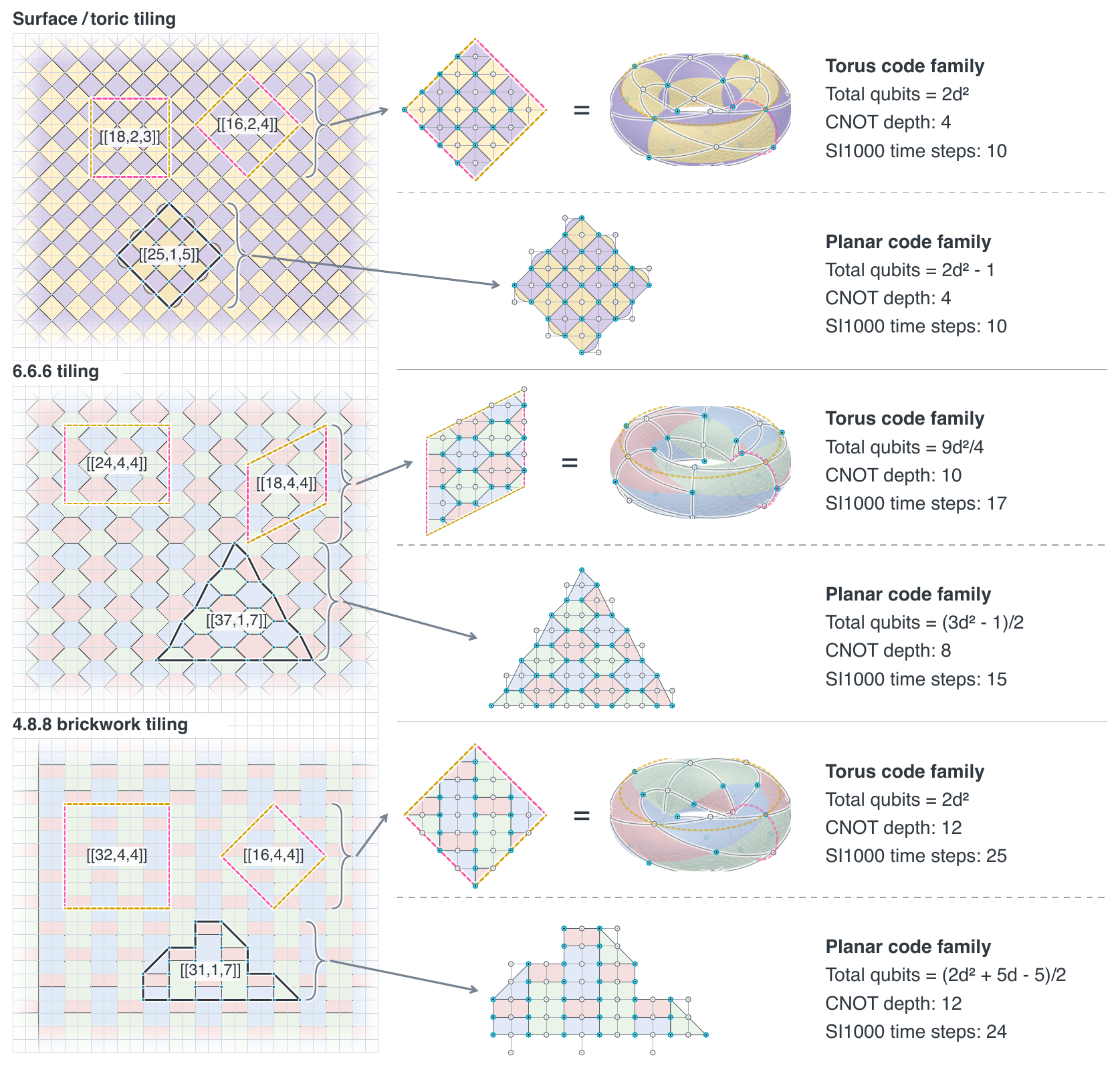}
% Detailed visual description retained for reference and future accessible editions.
% This source comment is not embedded screen-reader text in the compiled PDF.
% From stabilizer tilings to finite torus and planar codes.
% The fabrics share a square hardware grid: one edge is one nearest-neighbor coupler.
% Surface code coordinates are rotated and uniformly rescaled to match this physical
% pitch, not stabilizer-face width. The faint mesh continues between the rows,
% progressing from surface/toric through 6.6.6 to 4.8.8. Braces select the
% right-hand quotient and planar cutout; matching pink or gold edges are identified.
% Each equals sign connects flat and wrapped views of the same quotient.
% The planar surface code example has $d=5$; both color code examples have $d=7$.
% Solid contours follow retained check supports, including boundary truncations.
% Flat hardware views use dark check boundaries over light-gray circuit couplers.
% Blue-rimmed dots are data and gray-rimmed dots are ancillas.
% Tiling cutout labels give $[[n,k,d]]$. Resources give total SI1000 time
% steps per cycle, CNOT depth, and total data-plus-ancilla
% qubits as a function of d.
% Torus formulas apply to scaled
% multiples of these $d=4$ examples; planar 4.8.8 uses $d=4s-1$.
% All flat views preserve equal aspect ratio.
\caption{From stabilizer tilings to finite torus and planar codes.}
\label{fig:methods}
\end{figure}

\subsection{Superdense syndrome extraction}
The stabilizer checks of the 4.8.8 color codes are paired face operators
$S_f^X=\prod_{q\in f}X_q$ and $S_f^Z=\prod_{q\in f}Z_q$.
Bell-flagged extraction uses a Bell pair of ancillas to measure a face check
while retaining a flag outcome~\cite{bell_flagged}.
Superdense extraction acquires both check types in each complete cycle
by using the second ancilla outcome for the check in the other basis
instead of a flag~\cite{superdense}.
Unlike the 6.6.6 construction, which assigns an ancilla pair to each face,
our 4.8.8 circuits share ancillas between four- and eight-sided checks.
For simplicity, we describe the circuits in terms of CNOT gates.
Before simulation, we transpile these circuits to the native-CZ gate
set of the SI1000 noise model~\cite{si1000}.
Section~\ref{sec:native-implementation} details the implementation used in
Figs.~\ref{fig:feature} and~\ref{fig:comparison}.

\subsubsection{Measurement records and Pauli frames}
\label{sec:classical-construction}
Write the ancilla measurement record of one cycle as a binary column vector
$\mathbf{r}$, with outcome $r_j$ corresponding to eigenvalue $(-1)^{r_j}$.
All classical processing below consists of fixed XOR rules, written as
binary matrix products.

A \textbf{superdense frame update} is a tuple
$\mathcal{F}=(F^X,F^Z,\Lambda^X,\Lambda^Z)$ of binary matrices giving
the physical Pauli-frame bits and logical-sign updates:
\[
 \mathbf{x}=F^X\mathbf{r},\qquad \mathbf{z}=F^Z\mathbf{r},\qquad
 \boldsymbol{\lambda}^P=\Lambda^P\mathbf{r},\quad P\in\{X,Z\},
\]
where all matrix multiplication is modulo two.
Up to a global phase, the corresponding physical correction is
$\mathcal{P}=\prod_q X_q^{x_q}Z_q^{z_q}$.
The bit $\lambda_j^P$ records the sign flip of logical operator $L_j^P$
under this same correction.
The physical frame restores incoming face and logical signs and is
distinct from the decoder's correction for noise. We track it through
detector and logical-observable parities instead of applying feedback
gates~\cite{superdense,stim}.

Face-syndrome bits use the same notation:
$\mathbf{s}^{P,\mathrm{in}}=A^{P,\mathrm{in}}\mathbf{r}$ and
$\mathbf{s}^{P,\mathrm{out}}=A^{P,\mathrm{out}}\mathbf{r}$ give the signs
$(-1)^{s_f^{P,\mathrm{in/out}}}$ of $S_f^P$ before and after the ideal
cycle. A further matrix $C$ gives the ideal record constraints $C\mathbf{r}=\mathbf{0}$.

Tables~\ref{tab:torus-syndromes} and~\ref{tab:triangular-syndromes} list
each syndrome-parity row by the set of record indices at which it is one.
For example, $J=\{2,5\}$ gives $r_2\oplus r_5$, while $J=\varnothing$ gives zero.

We derive the superdense frame update matrices by introducing
\emph{reference qubits} into the calculation. To identify a
measurement-induced Pauli byproduct on a data qubit $d$, we first pair $d$
with an untouched reference qubit in the Bell state
$|\Phi^+\rangle=(|00\rangle+|11\rangle)/\sqrt{2}$.
The compensating Pauli is the operator on $d$ that restores this Bell state
after measurement. The first two gates below prepare the pair:
\[
\begin{quantikz}[row sep=0.3cm,column sep=0.4cm]
 \lstick{$\text{reference qubit}:\ |0\rangle$} & \gate{H} & \ctrl{1} & & & \\
 \lstick{$d:\ |0\rangle$} & & \targ{} & \targ{} & & \\
 \lstick{$a:\ |+\rangle$} & & & \ctrl{-1} & \meter{Z} & \rstick{$r$}\setwiretype{c}
\end{quantikz}
\]
The remaining CNOT and measurement leave
$(I\otimes X^r)|\Phi^+\rangle$: $X\otimes X$ still has sign $+1$,
while $Z\otimes Z$ has sign $(-1)^r$.
The correction $X_d^r$ restores the Bell state.
Here $x_d=r$, $z_d=0$, $\lambda^X=0$, and $\lambda^Z=r$.
Without the Bell preparation, the ancilla CNOT and measurement send an
arbitrary data input $|\psi\rangle$ to $X^r|\psi\rangle$.
For the full extraction circuit, we propagate the corresponding Pauli
correlations through the ideal Clifford gates and replace measured ancilla
operators by their recorded signs, obtaining the face-syndrome parities
and logical-sign updates.

Let $H$ contain the binary face supports, and let row $j$ of $L^P$ be
the binary support of logical operator $L_j^P$. A valid frame update satisfies
\[
\begin{aligned}
 HF^Z&=A^{X,\mathrm{in}}\oplus A^{X,\mathrm{out}},
 &L^X F^Z&=\Lambda^X,\\
 HF^X&=A^{Z,\mathrm{in}}\oplus A^{Z,\mathrm{out}},
 &L^Z F^X&=\Lambda^Z.
\end{aligned}
\]
The equations hold modulo the row space of $C$, since adding a constraint
does not change a parity on valid ideal records. Gaussian elimination gives
a valid representative frame matrix $F=\begin{pmatrix}F^X\\F^Z\end{pmatrix}$.

Memory experiments prepare and finally measure the data in the same Pauli
basis. As in standard Stim~\cite{stim} memory experiments, we add detector
annotations for products of consecutive measurements of the same face
check, accounting for the Pauli frame, with boundary detectors comparing
the first and last checks to data preparation and readout.
Independent ideal record constraints are also included.
Logical-observable annotations combine the final logical measurement with
the accumulated $\boldsymbol{\lambda}^P$ updates.

\subsubsection{Torus code circuit}
The torus circuit uses a mixed-polarity pattern: different ancilla groups use
opposite CNOT directions and corresponding preparation and measurement bases.
The same cycle repeats every round.
At the end of a memory experiment, the data share the final ancilla measurement step.

Figure~\ref{fig:torus-cycle} shows the complete spatial
repeat on the $[[16,4,4]]$ torus, with 16 ancillas and 120 CNOTs per cycle.
Tiling an $m\times m$ array and identifying its outer opposite edges gives
the $d=4m$ family, with $d^2$ data qubits, $d^2$ ancillas, and
$15d^2/2$ CNOTs in twelve layers per cycle.

\clearpage
\begin{figure}[H]
\centering\includegraphics[width=\linewidth]{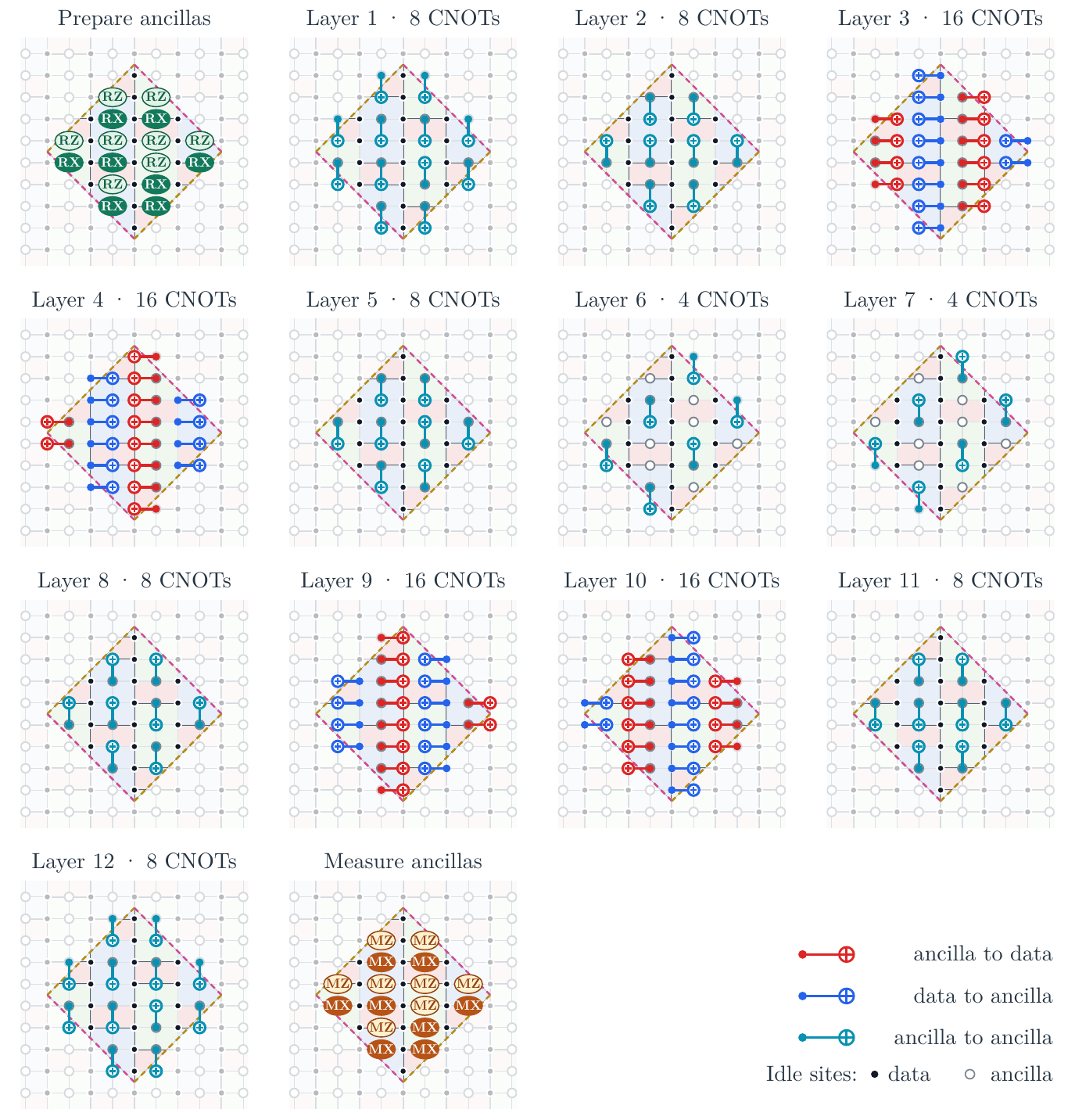}
\caption{Complete CNOT cycle for the 4.8.8 color code on a torus:
preparation, twelve CNOT layers, and measurement, read left to right and
then top to bottom. The cycle repeats unchanged each round.
Matching dashed seams are identified. Gates crossing the periodic boundary
are drawn twice. $\mathrm{RX}/\mathrm{MX}$ are implemented via a Hadamard
combined with $\mathrm{RZ}/\mathrm{MZ}$.
Long connections in the linked viewer denote periodic wraparound.
\TorusCrumbleLink}
\label{fig:torus-cycle}
\end{figure}

\clearpage
\paragraph{Syndrome parities and Pauli-frame updates.}
Tables~\ref{tab:torus-frame}--\ref{tab:torus-logical}
give the syndrome, frame, and logical-sign maps for the illustrated
$[[16,4,4]]$ circuit.
On the torus, the product of all same-Pauli face checks of any two colors
is identity. These relations provide known first-round parities even for
checks whose individual signs are not fixed by data preparation.

\begin{table}[H]
\centering
\begin{minipage}[c]{0.37\linewidth}
\centering\includegraphics[width=\linewidth]{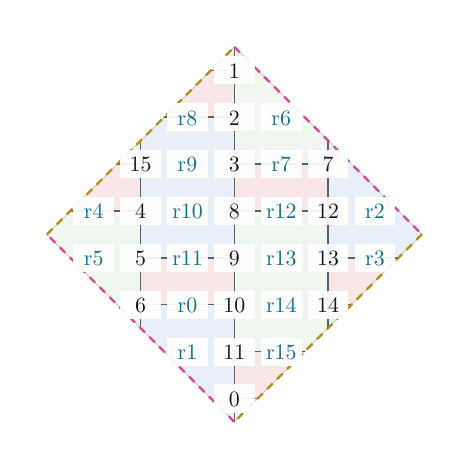}
\par\small Black: data.\quad Blue $r_j$: records.
\end{minipage}\hfill
\begin{minipage}[c]{0.61\linewidth}
\centering\small
\renewcommand{\arraystretch}{1.12}\setlength{\tabcolsep}{4pt}
% Generated by build_torus_frame.py; raw record indices are zero based.
\begin{tabular}{@{}ccc@{}}
\toprule
Record $j$ & $X$ support $S^X_j$ & $Z$ support $S^Z_j$\\
\midrule
$0$ & $\{6,10\}$ & $\varnothing$\\
$1$ & $\varnothing$ & $\{11,15\}$\\
$2$ & $\{6,7,10,11\}$ & $\varnothing$\\
$5$ & $\varnothing$ & $\{4,7,9,12,13,15\}$\\
$7$ & $\varnothing$ & $\{9,10,13,14\}$\\
$8$ & $\{4,6,7,10,12,15\}$ & $\varnothing$\\
$9$ & $\varnothing$ & $\{11,15\}$\\
$10$ & $\{4,6,9,11,12,13,14,15\}$ & $\varnothing$\\
$14$ & $\varnothing$ & $\{10,14\}$\\
\bottomrule
\end{tabular}

\end{minipage}
\caption{Torus Pauli frame. Each row contributes
$[X_{S_j^X}Z_{S_j^Z}]^{r_j}$, with $X_S=\prod_{q\in S}X_q$ and similarly
for $Z_S$.
The index key has the orientation and periodic identifications of
Fig.~\ref{fig:torus-cycle}.}
\label{tab:torus-frame}
\end{table}

\begin{table}[H]
\centering\small
\renewcommand{\arraystretch}{1.12}\setlength{\tabcolsep}{4pt}
% Generated by build_torus_frame.py; raw record indices are zero based.
\begin{tabular}{@{}cllcccc@{}}
\toprule
$f$ & Color & Data support $Q_f$ & $X_{\rm in}$ & $X_{\rm out}$ & $Z_{\rm in}$ & $Z_{\rm out}$\\
\midrule
$0$ & $\text{blue}$ & $\{0,1,6,7,10,11,12,13\}$ & $\{1,7\}$ & $\{1,15\}$ & $\{2\}$ & $\{2\}$\\
$1$ & $\text{green}$ & $\{0,1,2,3,4,5,6,7\}$ & $\{5\}$ & $\{5\}$ & $\{6,8\}$ & $\{0,6\}$\\
$2$ & $\text{red}$ & $\{5,6,9,10\}$ & $\{5,11\}$ & $\{11,14\}$ & $\{0,10\}$ & $\{0,10\}$\\
$3$ & $\text{red}$ & $\{3,7,8,12\}$ & $\{7\}$ & $\{7\}$ & $\{10,12\}$ & $\{2,12\}$\\
$4$ & $\text{blue}$ & $\{2,3,4,5,8,9,14,15\}$ & $\{9,15\}$ & $\{7,9\}$ & $\{10\}$ & $\{10\}$\\
$5$ & $\text{green}$ & $\{8,9,10,11,12,13,14,15\}$ & $\{14\}$ & $\{14\}$ & $\{0,13\}$ & $\{8,13\}$\\
$6$ & $\text{red}$ & $\{1,2,13,14\}$ & $\{3,14\}$ & $\{3,5\}$ & $\{2,8\}$ & $\{2,8\}$\\
$7$ & $\text{red}$ & $\{0,4,11,15\}$ & $\{1,7,9\}$ & $\{1,7,9\}$ & $\{2,4\}$ & $\{4,10\}$\\
\bottomrule
\end{tabular}

\caption{Complete torus face-sign maps. Each data support carries both an
$X$ and a $Z$ stabilizer.}
\label{tab:torus-syndromes}
\end{table}

\begin{table}[H]
\centering\small
\renewcommand{\arraystretch}{1.12}\setlength{\tabcolsep}{4pt}
\begin{minipage}[c]{0.60\linewidth}\centering
% Generated by build_torus_frame.py; raw record indices are zero based.
\begin{tabular}{@{}clclc@{}}
\toprule
$j$ & $\overline X_j$ support & $\lambda^X_j$ & $\overline Z_j$ support & $\lambda^Z_j$\\
\midrule
$0$ & $\{0,1,2,4\}$ & $r_{5}$ & $\{0,3,8,11\}$ & $r_{2}\oplus r_{10}$\\
$1$ & $\{1,2,5,6\}$ & $0$ & $\{3,5,8,9\}$ & $r_{10}$\\
$2$ & $\{3,5,8,9\}$ & $r_{5}\oplus r_{7}$ & $\{1,2,5,6\}$ & $r_{4}\oplus r_{12}$\\
$3$ & $\{0,3,8,11\}$ & $r_{1}\oplus r_{9}$ & $\{0,1,2,4\}$ & $r_{8}\oplus r_{10}$\\
\bottomrule
\end{tabular}

\end{minipage}\hfill
\begin{minipage}[c]{0.38\linewidth}\centering
% Generated by build_torus_frame.py; raw record indices are zero based.
\begin{tabular}{@{}cc@{}}
\toprule
$j$ & Ideal relation $g_j(\mathbf{r})=0$\\
\midrule
$0$ & $r_{1}\oplus r_{5}\oplus r_{7}\oplus r_{9}\oplus r_{14}\oplus r_{15}$\\
$1$ & $r_{0}\oplus r_{2}\oplus r_{6}\oplus r_{8}\oplus r_{10}\oplus r_{13}$\\
$2$ & $r_{0}\oplus r_{2}\oplus r_{4}\oplus r_{8}\oplus r_{10}\oplus r_{12}$\\
$3$ & $r_{1}\oplus r_{3}\oplus r_{9}\oplus r_{11}$\\
\bottomrule
\end{tabular}

\end{minipage}
\caption{Torus logical pairs and record constraints. Left: four logical
pairs and their per-cycle sign increments.
XOR $\lambda_j^P$ from each cycle into the final
data measurement of the indicated logical-$P$ support.
Right: four parity constraints, each ideally zero.}
\label{tab:torus-logical}
\end{table}

To scale the torus circuit to distance $d=4m$, tile the gate pattern in
Fig.~\ref{fig:torus-cycle} over an $m\times m$ array and identify opposite
outer edges. Recompute the logical-sign updates and record constraints in
Table~\ref{tab:torus-logical} for each distance using the binary linear-algebra
construction of Section~\ref{sec:classical-construction}.

\clearpage
\begin{figure}[H]
\centering\includegraphics[width=\linewidth]{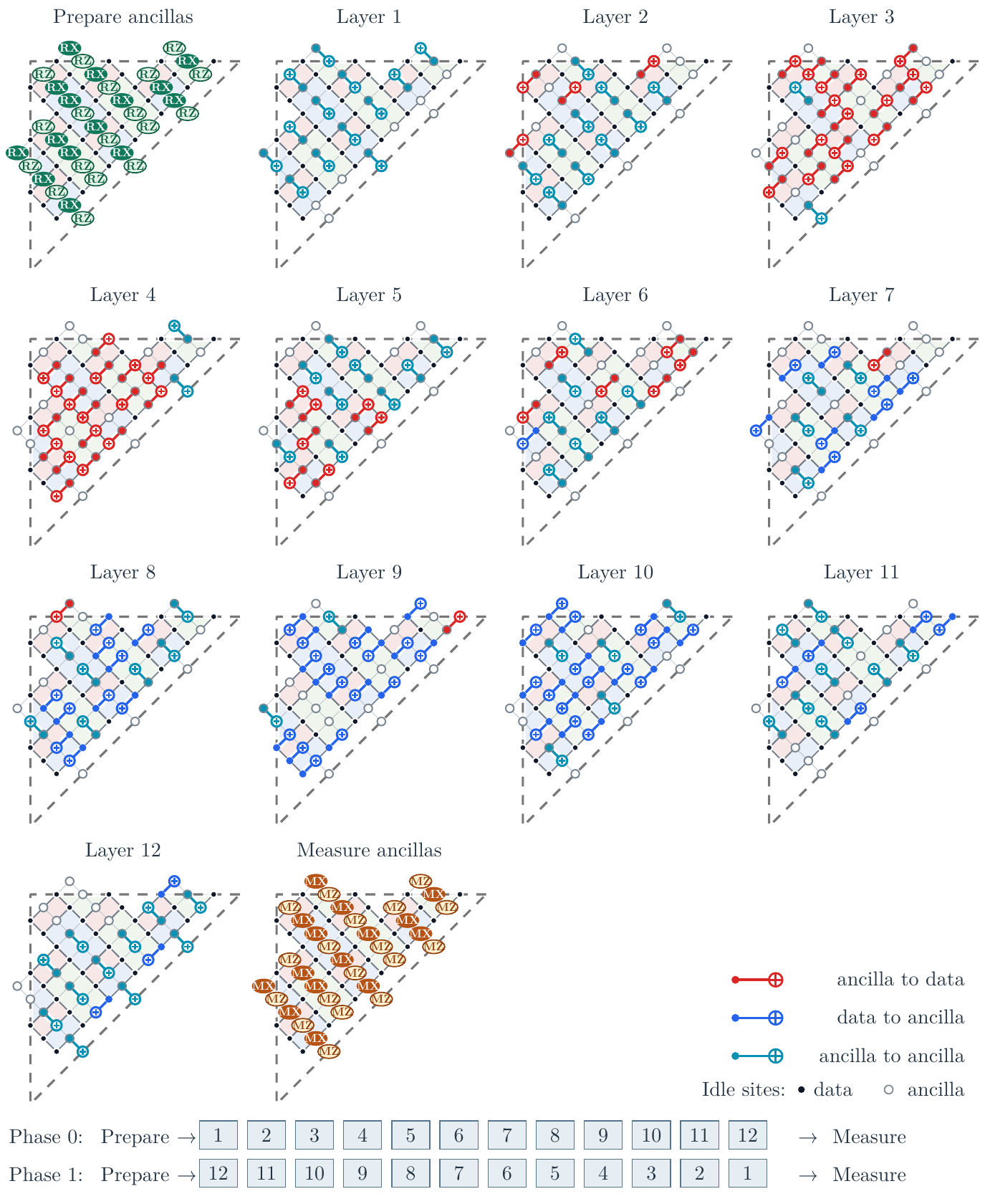}
\caption{Complete CNOT cycle at $d=7$ (31 data and 33 ancillas,
195 CNOTs per cycle). Read preparation, layers $L_1$--$L_{12}$, and
measurement left to right and then top to bottom for phase zero.
Phase one reverses the CNOT-layer order, not the control--target directions.
Both phases use the same preparation and measurement bases and extract
both $X$ and $Z$ checks.
\PlanarCrumbleLink}
\label{fig:cycle}
\end{figure}

\clearpage
\subsubsection{Triangular code circuit}
The triangular code circuit requires a few additional techniques to maintain
the code distance, described below.

\paragraph{Alternating forward--backward cycles to improve distance.}
Alternating a schedule with its reverse can mitigate hook-induced distance
loss, as also discussed for surface codes~\cite{diagonal_schedule}.
Write the first cycle's twelve parallel CNOT matchings as
$L_1,\ldots,L_{12}$. Ancillas are reset in $Z$, with Hadamards on a set
$H_{\rm in}$ before those layers and on $H_{\rm out}$ before terminal $Z$
measurement. The next cycle uses $L_{12},\ldots,L_1$ and exchanges
$H_{\rm in}$ with $H_{\rm out}$. The control and target of each individual
CNOT are unchanged. In the direct local rule shown here $H_{\rm in}=H_{\rm out}$,
so the ancilla basis pattern itself is the same in both phases.

\paragraph{A local scheduling rule with 90 equivalence classes.}
The triangular boundaries and corners introduce local ancilla environments
beyond those in the periodic bulk, requiring separate schedules for these
regions. For $d=4m-1$, we classify ancillas by their position within the
four-by-four bulk period and by the arrangement of data qubits, ancillas,
and vacant sites within Manhattan distance five, in lattice units.
These 90 local-environment classes include the boundary and small-size cases.
Each class specifies a local CNOT pattern and ancilla preparation and
measurement bases.
The rule uses twelve layers at every size, with its $d=7$ instance shown
in Fig.~\ref{fig:cycle}.
Appendix~\ref{app:schedule-spec} gives the coordinates and check supports.
The complete circuit generator, including all boundary and small-size cases,
is provided in the source bundle linked there.

\paragraph{Syndrome parities and Pauli-frame updates.}
The triangular syndrome and frame maps depend on the cycle phase $p=0,1$.
For $d=7$, the 33 ancilla bits contain 30 independent face-syndrome bits
and three individual flag records, $r_{10},r_{25},r_{30}$, each ideally zero.
All records are zero for $+1$ code-space input.
Phase 0 needs only a $Z$ frame, and phase 1
only an $X$ frame. Table~\ref{tab:frame} uses the same convention as
Table~\ref{tab:torus-frame}.

\begin{table}[H]
\centering
\begin{minipage}[c]{0.47\linewidth}
\centering
\includegraphics[width=\linewidth]{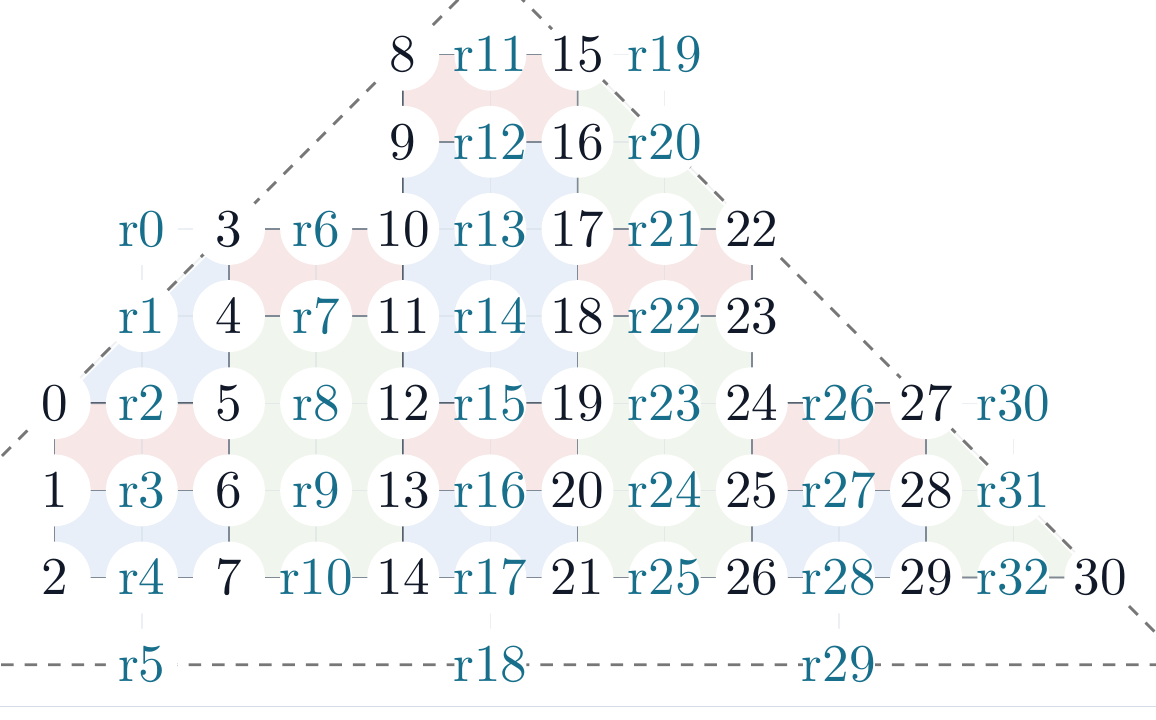}
\par\small Black numbers: data qubits.\quad Blue $r_j$: ancilla records.
\end{minipage}\hfill
\begin{minipage}[c]{0.50\linewidth}
\centering\small
\renewcommand{\arraystretch}{1.12}
\setlength{\tabcolsep}{4pt}
% Generated by build_frame.py; raw ancilla-record indices are zero based.
\begin{tabular}{@{}ccc@{}}
\toprule
Phase 0: $j$ & Phase 1: $j$ & Data support $S$\\
\midrule
$0$ & $1$ & $\{5,6,7,10,14,19,20\}$ \\
$2$ & $3$ & $\{10,11\}$ \\
$7$ & $6$ & $\{5,6,7,14,19,20\}$ \\
$8$ & $9$ & $\{6,7,20,21\}$ \\
$12$ & $11$ & $\{10,11,14,21,22,23\}$ \\
$13$ & $14$ & $\{11,14,19,20,23,26,29,30\}$ \\
$15$ & $16$ & $\{14,21\}$ \\
$19$ & $20$ & $\{16,19,20,21,26,29,30\}$ \\
$21$ & $22$ & $\{19,20,21,26,29,30\}$ \\
$23$ & $24$ & $\{20,21,28,29\}$ \\
$27$ & $26$ & $\{26,29\}$ \\
$31$ & $32$ & $\{29,30\}$ \\
\bottomrule
\end{tabular}

\end{minipage}
\caption{Complete nonidentity frame updates for the illustrated $d=7$ circuit.
Each row contributes $Z_S^{r_j}$ in phase 0 or $X_S^{r_j}$ in phase 1.
Omitted records contribute identity. The key uses the same geometry as
Fig.~\ref{fig:cycle}, rotated by $45^\circ$.}
\label{tab:frame}
\end{table}

For example, a lone $r_2=1$ gives $Z_{10}Z_{11}$ in phase 0, while a lone
$r_3=1$ gives $X_{10}X_{11}$ in phase 1.

\begin{table}[H]
\centering\small
\renewcommand{\arraystretch}{1.12}\setlength{\tabcolsep}{3pt}
% Generated by build_frame.py; each entry J denotes XOR of records r_j, j in J.
\begin{tabular}{@{}lcccccc@{}}
\toprule
 & \multicolumn{3}{c}{Phase 0} & \multicolumn{3}{c}{Phase 1}\\
\cmidrule(lr){2-4}\cmidrule(l){5-7}
Face support $S_f$ & $X_{\rm in}$ & $X_{\rm out}$ & $Z_{\rm in}=Z_{\rm out}$ & $X_{\rm in}=X_{\rm out}$ & $Z_{\rm in}$ & $Z_{\rm out}$\\
\midrule
$\{0,3,4,5\}$ & $\{0\}$ & $\{7\}$ & $\{1\}$ & $\{0\}$ & $\{1\}$ & $\{6\}$\\
$\{0,1,5,6\}$ & $\{2\}$ & $\{2,8\}$ & $\{3\}$ & $\{2\}$ & $\{3\}$ & $\{3,9\}$\\
$\{1,2,6,7\}$ & $\{4\}$ & $\{4\}$ & $\{5\}$ & $\{4\}$ & $\{5\}$ & $\{5\}$\\
$\{3,4,10,11\}$ & $\{7\}$ & $\{0,7,13\}$ & $\{6\}$ & $\{7\}$ & $\{6\}$ & $\{1,6,14\}$\\
$\{4,5,6,7,11,12,13,14\}$ & $\{8\}$ & $\{2,8,15\}$ & $\{9\}$ & $\{8\}$ & $\{9\}$ & $\{3,9,16\}$\\
$\{8,9,15,16\}$ & $\{12,13\}$ & $\{12,13,19\}$ & $\{11\}$ & $\{12\}$ & $\{11,14\}$ & $\{11,14,20\}$\\
$\{9,10,11,12,16,17,18,19\}$ & $\{13\}$ & $\{7,13,21\}$ & $\{14\}$ & $\{13\}$ & $\{14\}$ & $\{6,14,22\}$\\
$\{12,13,19,20\}$ & $\{15\}$ & $\{8,15,23\}$ & $\{16\}$ & $\{15\}$ & $\{16\}$ & $\{9,16,24\}$\\
$\{13,14,20,21\}$ & $\{17\}$ & $\{17\}$ & $\{18\}$ & $\{17\}$ & $\{18\}$ & $\{18\}$\\
$\{15,16,17,22\}$ & $\{19\}$ & $\{12\}$ & $\{20\}$ & $\{19\}$ & $\{20\}$ & $\{11\}$\\
$\{17,18,22,23\}$ & $\{21\}$ & $\{13,21\}$ & $\{22\}$ & $\{21\}$ & $\{22\}$ & $\{14,22\}$\\
$\{18,19,20,21,23,24,25,26\}$ & $\{23\}$ & $\{15,23,27\}$ & $\{24\}$ & $\{23\}$ & $\{24\}$ & $\{16,24,26\}$\\
$\{24,25,27,28\}$ & $\{27\}$ & $\{23,27\}$ & $\{26\}$ & $\{27\}$ & $\{26\}$ & $\{24,26\}$\\
$\{25,26,28,29\}$ & $\{28\}$ & $\{28,31\}$ & $\{29\}$ & $\{28\}$ & $\{29\}$ & $\{29,32\}$\\
$\{27,28,29,30\}$ & $\{31\}$ & $\{27,31\}$ & $\{32\}$ & $\{31\}$ & $\{32\}$ & $\{26,32\}$\\
\bottomrule
\end{tabular}

\caption{Complete triangular face-sign maps at $d=7$, with data and record
labels from Table~\ref{tab:frame}. Columns with equal incoming and outgoing
signs are combined.}
\label{tab:triangular-syndromes}
\end{table}

\begin{samepage}
The logical pair and its phase-dependent increments (left) and ideal-zero
flag relations (right) use the same conventions as Table~\ref{tab:torus-logical}.
The flag relations are identical in both phases.
\begin{center}\small
\renewcommand{\arraystretch}{1.12}\setlength{\tabcolsep}{4pt}
\begin{minipage}[c]{0.73\linewidth}\centering
% Generated by build_frame.py; raw record indices are zero based.
\begin{tabular}{@{}clcc@{}}
\toprule
Phase $p$ & Common $\overline X,\overline Z$ support $\ell$ & $\lambda_p^X$ & $\lambda_p^Z$\\
\midrule
$0$ & $\{0,1,2,3,8,9,10\}$ & $r_{0}\oplus r_{2}\oplus r_{12}$ & $0$\\
$1$ & $\{0,1,2,3,8,9,10\}$ & $0$ & $r_{1}\oplus r_{3}\oplus r_{11}$\\
\bottomrule
\end{tabular}

\end{minipage}\hfill
\begin{minipage}[c]{0.25\linewidth}\centering
% Generated by build_frame.py; applies in both phases.
\begin{tabular}{@{}cc@{}}
\toprule
$j$ & Ideal relation $g_j(\mathbf{r})=0$\\
\midrule
$0$ & $r_{30}$\\
$1$ & $r_{25}$\\
$2$ & $r_{10}$\\
\bottomrule
\end{tabular}

\end{minipage}
\end{center}
\end{samepage}

\subsection{Native-CZ implementation}
\label{sec:native-implementation}
To match the SI1000 gate set~\cite{si1000}, the numerical results use native
CZ gates.
The identity $\operatorname{CNOT}_{c,t}=H_t\operatorname{CZ}_{c,t}H_t$
connects the two descriptions, but compiling each gate separately is
unnecessarily costly: adjacent rotations can cancel, and operations on
disjoint qubits can overlap. Table~\ref{tab:native-costs} gives the costs
of the implementations actually simulated. A circuit moment is one scheduled
time step, in the terminology of SI1000, and may contain simultaneous
operations on disjoint qubits. Figure~\ref{fig:methods} gives the CNOT depth
of the explanatory construction and the total SI1000 time-step count of
the simulated native circuit, which also includes
single-qubit rotations, reset, and measurement. Time steps need not have
equal hardware duration: reset and measurement incur their own error rates
and additional noise on the other qubits (Appendix~\ref{app:simulation}).

\begin{table}[H]
\centering\small
\renewcommand{\arraystretch}{1.18}
\begin{tabular}{@{}lrr@{}}
\toprule
Code & CNOT depth & SI1000 time steps\\
\midrule
\multicolumn{3}{@{}l}{\textit{Planar}}\\
Surface code & 4 & 10\\
6.6.6 color code & 8 & 15\\
4.8.8 color code & 12 & 24\\
\midrule
\multicolumn{3}{@{}l}{\textit{Torus}}\\
Toric code & 4 & 10\\
6.6.6 color code (both quotients) & 10 & 17\\
4.8.8 color code & 12 & 25\\
\bottomrule
\end{tabular}
\caption{Per-cycle costs. The numerical columns give the CNOT-layer count
and the native time-step count, respectively.
Triangular 4.8.8 uses 23 time steps at $d=3$ and 24 at $d\ge7$.
Memory initialization and final measurement
are included in simulations but not in the recurring-cycle counts.}
\label{tab:native-costs}
\end{table}

For surface and 6.6.6 color codes, we compile the CNOT constructions into
CZ interactions with local Hadamards, canceling compatible rotations before
physically lowering $X$ reset and measurement. The torus 6.6.6 comparison
uses the ten-layer CNOT schedule (Appendix~\ref{app:reference-circuits}).

For 4.8.8, the native optimization additionally changes boundary Bell-pair
preparation polarity in the triangle and interior ancilla Clifford rotations
on the torus. Both retain the same entangling-gate count and hardware edges
as the illustrated construction. The native operations are scheduled as
early as their order on each individual qubit allows, followed by a common
ancilla measurement. Packing can stagger gates from one original CNOT layer
across several time steps, producing twenty CZ-bearing steps.

The native circuits also use local Hadamard frames on data and different
ancilla-record conventions. Preparation and final measurement are in the corresponding
physical bases, and the face, flag and logical parities are transformed
consistently. Appendix~\ref{app:native-spec} gives the native scheduling rule
and explicit conversions for the classical tables above. Signed checks of
the incoming and outgoing stabilizers and all logical generators verify
this connection. Circuit-distance evidence additionally accounts for native
fault locations.

\section{Results}
\label{sec:results}
Figure~\ref{fig:comparison} compares \CZTopologyText{} memories from the
surface and color code families using native CZ gates under the SI1000
noise model~\cite{si1000}. The circuits are described in
Section~\ref{sec:methods}.
Simulation, decoder, and statistical
methods are given in Appendix~\ref{app:simulation}.

% Keep the Results opening and comparison together in the available page space.
\begingroup
\setlength{\intextsep}{4pt}
\captionsetup{skip=4pt}
\begin{figure}[H]
\centering\includegraphics[width=0.88\linewidth,trim=0 4bp 0 14bp,clip]{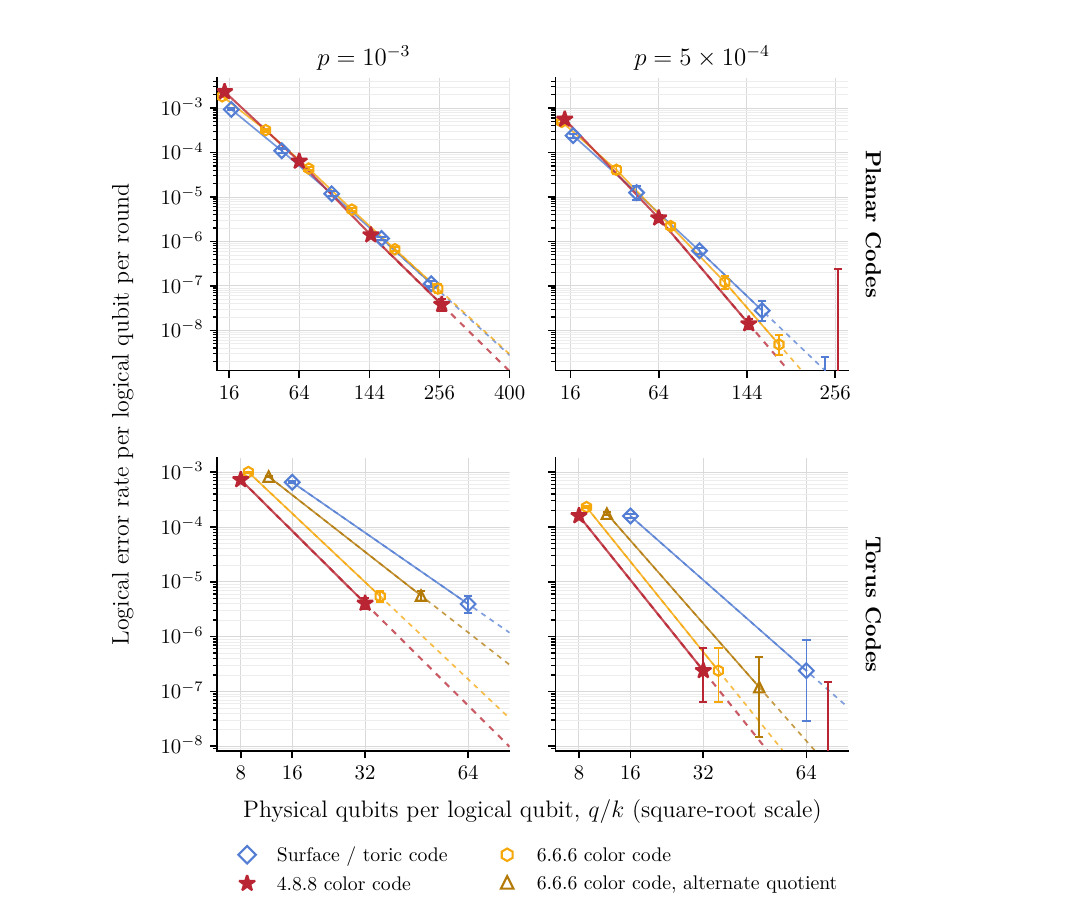}
% Visual description retained for reference and future accessible editions.
% This source comment is not embedded screen-reader text in the compiled PDF.
% Red stars denote planar 4.8.8 color codes; blue denotes surface codes;
% gold denotes 6.6.6 color codes. Both panels use native CZ gates and SI1000 noise.
\caption{Comparison of quantum memories under SI1000 noise with native CZ gates,
at $p=10^{-3}$ (left) and $p=5\times10^{-4}$ (right).
The simulations contribute \CZTotalShots{} shots. Each point pools shots and errors
from the $X$- and $Z$-memory experiments.
Error bars represent 95\% confidence intervals.
The horizontal axis counts total physical qubits per logical
qubit on a square-root scale.
Dashed lines extrapolate $\log\ell$ linearly in $\sqrt{q/k}$ through the two
largest sizes with observed errors. Decoding uses
Tesseract~\cite{tesseract} with parameters described
in Appendix~\ref{app:simulation}.}
\label{fig:comparison}
\end{figure}
\endgroup

We count every touched physical qubit, including extraction ancillas:
$q=n_D+n_A$. The code distance is $d$, and the number of logical qubits
is $k$. We use $d_{\rm circ}$ for the minimum number of
physical fault locations producing an undetected logical error.

\begin{table}[H]
\centering\small
\renewcommand{\arraystretch}{1.25}
\begin{tabularx}{\linewidth}{@{}lcccX@{}}
\toprule
Code / extraction layout & $n_D$ & $q$ & $k$ & Circuit-distance evidence\\
\midrule
\multicolumn{5}{@{}l}{\textit{Planar triangular / rotated patches}}\\
Surface code & $d^2$ & $2d^2-1$ & 1 & Hook-avoiding distance-$d$ schedule~\cite{surface_hook_schedule}\\
6.6.6 color code & $(3d^2+1)/4$ & $(3d^2-1)/2$ & 1 & Reference distance-$d$ construction~\cite{superdense}\\
4.8.8 color code & $(d^2+2d-1)/2$ & $(2d^2+5d-5)/2$ & 1 & Checked at distances $3,7,11,15$\\
\midrule
\multicolumn{5}{@{}l}{\textit{Torus quotients}}\\
Toric code & $d^2$ & $2d^2$ & 2 & Standard distance-$d$ schedule\\
4.8.8 color code & $d^2$ & $2d^2$ & 4 & Checked at distances $4,8$. Upper bound $d_{\rm circ}\leq12$ at $d=12$.\\
6.6.6 color code (best) & $9d^2/8$ & $9d^2/4$ & 4 & Checked at distance $4$ (CNOT construction)\\
6.6.6 color code (alternate) & $3d^2/2$ & $3d^2$ & 4 & Checked at distance $4$ (CNOT construction)\\
\bottomrule
\end{tabularx}
\caption{Resource polynomials for the leading layouts and comparison families.
Torus color code scalings refer to multiples of the $d=4$ representatives.
The triangular 4.8.8 family uses $d=4s-1$.
``Checked'' means matching lower and upper bounds establish $d_{\rm circ}=d$
for $r=d$ extraction cycles in both memory bases and, for triangular 4.8.8,
both starting phases.}
\label{tab:resources}
\end{table}

Triangular 4.8.8 circuits use $60m^2-20m-5$ entangling gates per cycle
for $d=4m-1$, in either representation. At $d=7,11,15$, some qubits
participate in twelve CNOTs per cycle, so scheduling these gates requires
at least twelve layers.

\paragraph{Torus-to-planar overhead.}
As seen in Table~\ref{tab:resources}, the qubit overhead per logical qubit
at matching distance typically doubles asymptotically when moving from
torus to planar codes. The standard 6.6.6 triangle is an exception: its
overhead is $8/3$ times that of the best torus quotient, but twice that of
the torus quotient with larger overhead. Forming a triangular patch requires three
colored boundaries and modified checks at its edges. Changing these checks
can change which Pauli errors act as logical operators and their minimum
weight, so opening a compact torus need not preserve its distance or
overhead.

\section{Discussion and future work}
The 4.8.8 color code reduces the asymptotic qubit-overhead coefficient:
including extraction ancillas, $q/(kd^2)$ approaches $1$ for the planar family,
compared with $3/2$ for the 6.6.6 color code and $2$ for the surface code.
Under the standard subthreshold scaling model, Appendix~\ref{app:crossover}
shows that a smaller coefficient translates into a lower qubit cost at
sufficiently low physical error rates and sufficiently demanding logical-error
targets, even if the subthreshold decay constant is less favorable. The
practical question is where this crossover occurs. Our circuit-level simulations show favorable performance
already at $p=10^{-3}$ and $5\times10^{-4}$
(Figs.~\ref{fig:feature} and~\ref{fig:comparison}).
We conclude with suggested directions for extension and improvement.

\paragraph{Circuit optimization.}
Circuit depth and ancilla count remain targets for optimization. For example, we
found ten-layer CNOT circuits for the $[[64,4,8]]$ torus that correctly
extract the stabilizers in the absence of noise, compared with the
twelve-layer construction used here. However, we were unable to find a
shallower circuit that preserved the full code distance.
Moreover, some ancillas at the boundary of the triangular circuit produce
flag records rather than independent stabilizer syndromes and could
potentially be removed if a suitable distance-preserving extraction schedule
were found.
The search should also consider more compact circuits that sacrifice some
circuit distance because at practical physical error rates, they may achieve lower
logical error rates.

% Keep the short decoding paragraph and its footnote on this page.
\enlargethispage{\baselineskip}
\paragraph{Faster decoding.}
We made little attempt to optimize Tesseract's decoder parameters and did
not use error sparsification or multipass
decoding.\footnote{See the Tesseract pull requests for
\href{https://github.com/quantumlib/tesseract-decoder/pull/254}{error sparsification}
and \href{https://github.com/quantumlib/tesseract-decoder/pull/256}{multipass decoding}.}
Tuning these parameters and exploring
these features could improve the decoding speed--accuracy tradeoff.

The logical error rate scaling we observe with Tesseract provides a baseline
for developing faster near-optimal decoders for our 4.8.8 code memories.
In particular, neural-network decoders~\cite{bell_flagged,lacroix,alphaqubit2}
and pre-decoders~\cite{color_predecoder} have been demonstrated for 6.6.6
color codes. It may be possible to apply these methods to our 4.8.8
color codes while retaining their favorable logical error rate scaling.

\paragraph{Logical operations and denser memories.}
A next step is to implement and benchmark logical Clifford gates and lattice
surgery between our 4.8.8 patches, other color codes, and surface codes,
building on color-code lattice surgery and fault-tolerant code
interfaces~\cite{color_lattice_surgery,code_interface}.
Logical parity measurements could also enable yoking, in which an outer
code protects several logical memories using lattice-surgery
checks~\cite{yoked}.

Another possibility is to encode several logical qubits in a dense planar
layout, as in stellated color and surface codes~\cite{boundaries},
crosshair codes~\cite{crosshair}, and twist-defect packing~\cite{denser_surface}.
For our 4.8.8 memories, an open task is to develop distance-preserving
extraction circuits for such layouts on square-lattice hardware and evaluate
their overhead and logical operations.

\paragraph{Intermediate code distances.}
Our triangular 4.8.8 circuit construction covers $d=4m-1=3,7,11,\ldots$,
whereas the standard rotated surface and triangular 6.6.6 families cover
every odd distance. An open question is whether alternative boundaries and
extraction schedules can realize the other odd-distance class,
$d=4m+1=5,9,13,\ldots$, in this brickwork layout while retaining comparable
qubit overhead and preserving code distance. An analogous question is to
extend our $d=4m$ torus circuit family to intermediate even distances such as
$6,10,\ldots$, using other periodic identifications and compatible local
extraction schedules.

\paragraph{Denser planar 6.6.6 color codes.}
A concrete open problem is to recover the best 6.6.6 torus density in a
planar $k=1$ family with only the factor-two overhead cost:
$n_D=(9/16)d^2+O(d)$, and $q=(9/8)d^2+O(d)$ with superdense extraction,
instead of the usual $q\sim(3/2)d^2$.
Translating this periodic density into a planar construction requires
careful construction of the boundaries.

\paragraph{Exotic stabilizer tilings.}
Square-lattice hardware can support overlapping or nonconvex stabilizer regions.
Larger unit cells and more elaborate extraction circuits further broaden the
search space, partially explored in~\cite{tilecodes,nearestneighbor}. Local check
support alone does not guarantee nearest-neighbor extraction, and extra
circuit depth must earn its cost in gate and idle noise, as routed tile-code
implementations illustrate~\cite{routedtiles}. The program of
Fig.~\ref{fig:methods} can extend beyond face tilings: start from periodic
checks, construct admissible quotients and boundaries, then synthesize and
compare complete noisy memories.

\paragraph{Dynamic and Floquet codes.}
Time-dependent codes enlarge this space further. Floquet codes protect logical
information through an evolving sequence of stabilizer groups~\cite{floquet}.
Extending the search to periodic code deformations and measurement schedules
would require tracking logical information and fault propagation throughout
the full cycle, rather than optimizing extraction for a single fixed code.

\paragraph{Alternative hardware geometries.}
The research program shown in Figure~\ref{fig:methods} could be
carried out on a six-neighbor triangular lattice, a three-neighbor
honeycomb lattice, or a heavy-hex graph with degree-two and degree-three
vertices~\cite{heavyhex}.

\subsection*{Acknowledgments}
I thank Guang Hao Low, Oscar Higgott, Craig Gidney, Dripto M. Debroy, Laleh Aghababaie Beni, Emma Rosenfeld, and Patrick M. Harrington for helpful conversations.
I thank Cody Jones for suggesting to fairly benchmark the scaling of different topological codes using near-optimal decoders.
I also thank Fionn Malone for computational support and Hartmut Neven for creating an environment where this work is possible.

\subsection*{AI acknowledgment}
We used ChatGPT, Claude, and Gemini to help develop software for constructing,
optimizing, simulating, and visualizing the quantum memories studied here
and to prepare the manuscript. Recent improvements in AI capabilities have
dramatically improved these systems' ability to optimize quantum circuits
using tools such as Stim~\cite{stim}, although their scientific writing still has
room for improvement.

\begingroup\small
\makeatletter
\renewcommand{\@biblabel}[1]{[#1]}
\makeatother
\bibliographystyle{rainbowalpha}
\bibliography{references}
\endgroup

\appendix
\section{Simulation and statistical methods}
\label{app:simulation}
We pool shots and errors across the $X$ and $Z$ memories, including both
starting schedules for triangular 4.8.8. The $X$ and $Z$ shot counts differ
by at most \CZMaxBasisShotImbalancePercent\% relative to the smaller count.

A shot fails if any logical observable is decoded incorrectly. From the
fraction of failed shots $P$, we report the error rate per logical qubit
per round,
\[
 \ell=1-(1-P)^{1/(kr)},
\]
where $k$ is the number of logical qubits and $r=d$ is the number of rounds.
Error bars show 95\% Clopper--Pearson confidence intervals converted to
the same scale.

Dashed lines extend the trend through the two largest code sizes with
observed errors, assuming $\log\ell$ decreases linearly with $\sqrt{q/k}$.

\begin{minipage}{\linewidth}
The simulations use Stim~1.16~\cite{stim} and
Tesseract~\cite{tesseract}, pinned to source
\texttt{024db1d3b5b0}. Here is an example invocation of Tesseract
demonstrating the settings:
\begin{lstlisting}[basicstyle=\small\ttfamily,columns=fullflexible,
  keepspaces=true,showstringspaces=false,breaklines=true,
  frame=single,rulecolor=\color{black!20}]
./tesseract --circuit circuit.stim --dem decoder.dem \
  --sample-num-shots N \
  --stats-out stats.json --threads 2 \
  --num-det-orders 21 \
  --det-order-bfs --det-order-coordinate --det-order-index \
  --det-order-seed 518278944 --beam 20 \
  --beam-climbing --no-revisit-dets --pqlimit 1000000
\end{lstlisting}
\end{minipage}

We apply SI1000 noise~\cite{si1000} to the scheduled native circuits:
depolarization $p$ after CZ gates and $p/10$ after single-qubit Cliffords
and on idle qubits, reset error $2p$, and premeasurement error $5p$.
During reset or measurement, every noncollapsed qubit receives additional
independent depolarization $2p$.
Each single-qubit Clifford occupies one time step, including the Hadamards
used with $Z$ reset and measurement to implement $X$-basis preparation and readout.

The 4.8.8 simulations use equal fixed shot targets across memory bases and
starting phases at each distance and physical error rate.

\section{6.6.6 color code circuits}
\label{app:reference-circuits}
For triangular 6.6.6 memories, we use the standard eight-layer superdense
circuit~\cite{superdense} shown in Fig.~\ref{fig:hex-schedules}(a).
For torus memories, we use the ten-layer variant in
Fig.~\ref{fig:hex-schedules}(b), which preserves circuit distance in both
$d=4$ torus codes checked in Table~\ref{tab:resources}.
All 6.6.6 torus simulations in this paper use the native-CZ compilation
of this ten-layer circuit.

\begin{figure}[htbp]
\centering
\input{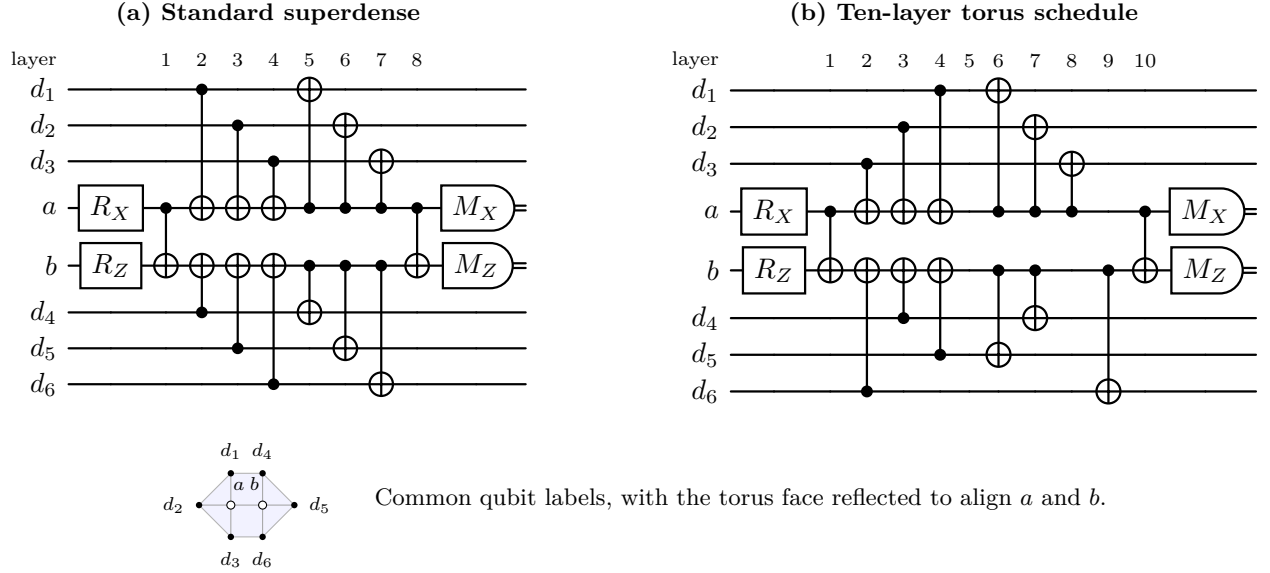}
\caption{Superdense extraction on one 6.6.6 hexagon.
(a) The standard eight-layer circuit~\cite[Fig.~2]{superdense}, used for all bulk face colors.
(b) A blue plaquette of the ten-layer circuit used for the torus memories.
The same fourteen CNOTs are ordered differently. Red and green torus plaquettes
use the same interaction pattern with different contact timings, specified in
the source circuits. Column numbers label the global
CNOT layers. Only gates involving the displayed ancilla pair are shown,
so a blank on a data wire may correspond to an interaction with a neighboring ancilla.}
\label{fig:hex-schedules}
\end{figure}

% Adapted from the small-code catalog's overhead argument. Use the iterated
% limit: an unrestricted joint limit ignores the minimum admissible distance.
\section{Asymptotic overhead and the low-noise crossover}
\label{app:crossover}
A lower qubit-overhead coefficient can compensate for weaker subthreshold
error suppression.
A related calculation comparing color- and surface-code lattice surgery at
matched logical error rates appears in Ref.~\cite[Sec.~IV~B]{color_lattice_surgery}.
For memories compared at the same physical noise parameter and logical-error
metric, assume for each family $F$
\begin{equation}
 \ell_F(d,p)=A_F(c_F p)^{d/2},
 \qquad
 \frac{q_F(d)}{k_F}=\Omega_F d^2+O(d),
 \label{eq:crossover-model}
\end{equation}
with fixed $A_F>0$, $c_F>0$, and $\Omega_F>0$, in a sufficiently low-noise
regime where $c_F p<1$. The low-noise decay constant $c_F$ depends on the
code, extraction circuit, noise model, and decoder, and can differ from
the inverse threshold.
Here $\ell_F$ is the logical error rate per logical qubit per round, and $q_F$
includes extraction ancillas. The exponent $d/2$ describes the leading
error-correcting power of distance-preserving extraction.
Admissible distances
are unbounded with bounded gaps, allowing odd distances, $d=4s-1$, or
integer multiples of four.

For a target $\varepsilon$, define the minimum per-logical-qubit cost
\[
 Q_F(\varepsilon,p)=\min_{d:\,\ell_F(d,p)\leq\varepsilon}
 \frac{q_F(d)}{k_F}.
\]
\begin{proposition}[Low-noise crossover]
For two families satisfying Eq.~\eqref{eq:crossover-model} and fixed
$0<p<\min(c_F^{-1},c_G^{-1})$ within the common validity range of the model,
\begin{equation}
 \lim_{\varepsilon\to0^+}\frac{Q_F(\varepsilon,p)}{Q_G(\varepsilon,p)}
 =\frac{\Omega_F}{\Omega_G}
 \left[\frac{\ln[1/(c_G p)]}
 {\ln[1/(c_F p)]}\right]^2.
 \label{eq:crossover-ratio}
\end{equation}
Consequently,
\[
 \lim_{p\to0^+}\lim_{\varepsilon\to0^+}
 \frac{Q_F(\varepsilon,p)}{Q_G(\varepsilon,p)}
 =\frac{\Omega_F}{\Omega_G}.
\]
If $\Omega_F<\Omega_G$, then for every sufficiently small fixed $p>0$,
family $F$ costs fewer qubits for all sufficiently small targets
$\varepsilon<\varepsilon_\star(p)$.
\end{proposition}

\begin{proof}
Inverting the error model gives the continuous distance requirement
\[
 D_F(\varepsilon,p)=
 \frac{2\ln(A_F/\varepsilon)}{\ln[1/(c_F p)]}.
\]
At fixed $p$, this grows without bound as $\varepsilon\to0$.
Every eligible distance is at least $D_F$, while rounding up to an admissible
distance costs only $O(1)$. The quadratic footprint then bounds the minimum
cost from both sides, giving $Q_F=\Omega_F D_F^2+O(D_F)$.
Since $\ln(A_F/\varepsilon)/\ln(A_G/\varepsilon)\to1$,
the cost ratio tends to Eq.~\eqref{eq:crossover-ratio}.
Finally, $\ln[1/(c_G p)]/\ln[1/(c_F p)]\to1$
as $p\to0$. If $\Omega_F/\Omega_G<1$, the limiting ratio is below one
for sufficiently small fixed $p$, and so is the finite-target ratio for
sufficiently small $\varepsilon$.
\end{proof}

\section{Triangular code geometry}
\label{app:schedule-spec}
% Generated by build_schedule_spec.py; edit its source, not this file.
\paragraph{Coordinates and checks.}
The following integer coordinates specify the triangular family at any
$d=4m-1$, $m\geq1$. Set $R=-4m-1$ and
\[
 L_j=R-\min(2j+3,\,8m-2j-3)+1.
\]
Data qubits occupy $(x,6+2j)$ for $0\leq j\leq4m-2$ and
$L_j\leq x\leq R$. Ancillas occupy $(x,7+2j)$ for
$0\leq j\leq4m-3$, with
\[
\begin{array}{c|cc}
 &\text{left endpoint}&\text{right endpoint}\\\hline
 j\ \text{even}&L_{j+1}&R+1\\
 j\ \text{odd}&L_j&R
\end{array}
\]
and every integer $x$ between these endpoints. Only nearest-neighbor
couplers between occupied sites are needed. The triangular record and
Pauli-frame tables number data qubits and ancillas separately from zero,
ordered first by increasing $y$ and then by increasing $x$.
These source axes have north equal to $+y$ and east equal to $+x$.
The triangular circuit figures use the rigid rotation
$(x,y)\mapsto(y,-x)$, with all gates and qubit indices unchanged.

The check supports are translates of
\[
\begin{aligned}
 Q&=\{0,1\}\mathbin\times\{0,2\}, &
 O&=\{0,1,2,3\}\mathbin\times\{0,2\},\\
 U&=\{(0,2),(1,2),(2,0),(2,2)\},&
 V&=\{(0,0),(1,0),(2,0),(2,2)\}.
\end{aligned}
\]
Include $Q+(x,y)$ at every data site with $x-y\equiv3\pmod4$
when the translated support is wholly present, and also at
$(x,y)=(R-1,6+4h)$ for $0\leq h\leq2m-2$.
Include $O+(x,y)$ when $x$ is even, $x\equiv y\pmod4$, and its
support is wholly present. Finally, include
$U+(-4m-5-4h,6+4h)$ for $0\leq h\leq m-2$ and
$V+(-8m+1+4h,4m+4+4h)$ for $0\leq h\leq m-1$.
Each support supplies both an $X$ and a $Z$ check.
Order faces by their minimum $y$, minimum $x$, support size, and finally
the lexicographically sorted coordinate list, breaking ties in that order.
Products of $X$, or of $Z$, over all data qubits represent logical $X$ and $Z$.

The complete distance-parameterized circuit generator, including the
boundary schedules, and the Stim circuits used in the simulations are
available in the
\href{https://doi.org/10.5281/zenodo.22820614}{Zenodo source bundle}.

\section{Native-CZ schedules and classical conventions}
\label{app:native-spec}
We use the superdense frame update and syndrome-parity matrices of
Section~\ref{sec:classical-construction}. Write the native ancilla record
as $\mathbf{w}$ and the corresponding CNOT-table record as $\mathbf{r}$, both in ascending
ancilla-index order. A memory contains $T$ cycles.

\subsection{Scheduling and physical bases}
\label{sec:native-packing}
Within each cycle, omit identity rotations and schedule each gate as early
as possible while preserving the operation order on every qubit. Measure
all ancillas together immediately after the final gate.

In a data Hadamard frame $U$, a Pauli $S$ from the CNOT description
is represented physically as $USU^\dagger$. This preserves the Pauli support
and the code distance.
For a logical-$P$ memory,
prepare and finally measure each data qubit in the physical basis specified
by the input and output frames, using $Z$ reset/measurement and actual
Hadamards where required. Intermediate frame choices are already implemented
by the native gate sequence.

\subsection{Triangular 4.8.8 circuit}
Begin with the forward CNOT schedule from the source bundle linked in
Appendix~\ref{app:schedule-spec}.
Let $d=4m-1$, and let $x_R=-(4m+1)$ be its rightmost data coordinate.
For each first-layer Bell pair
\[
 a=(x_R-2,9+4j),\quad b=(x_R-1,9+4j),\qquad
 j=0,\ldots,2m-3,
\]
replace $\operatorname{CNOT}_{a,b}$ by
$\operatorname{CNOT}_{b,a}$ and exchange the two preparation bases.
The set is empty at $d=3$. This replaces preparation from $|+0\rangle$
by preparation from $|0+\rangle$ and produces the same Bell state.
Define phase 1 by reversing the modified CNOT layers and exchanging its
preparation and measurement bases.

Let $D$ be the set of all data qubits. The input/output data frames are
$(H_D,I)$ in phase 0 and $(I,H_D)$ in phase 1. After resetting the ancillas,
initialize a pending set $B$ to the symmetric difference of the preparation
Hadamard set and the input data frame. For each CNOT matching, let $\mathcal{T}$ be its
targets and $A$ all its endpoints. Apply $H$ on $(B\mathbin\triangle\mathcal{T})\cap A$,
apply CZ on the matching's edges, and replace
$B\leftarrow(B\setminus A)\cup\mathcal{T}$. After the last matching, apply $H$ on
$B\mathbin\triangle H_{\rm measure}\mathbin\triangle H_{\rm output}$.
Finally pack the operations by Section~\ref{sec:native-packing}.
This retains the original entangling count while giving the 24-moment
cycle (23 at $d=3$).

\paragraph{Frame and syndrome conversion.}
In phase 0, substitute $r_j=w_j$. In phase 1, exchange the records within
each Bell pair reversed above. At $d=7$ this is
\[
 r_8=w_9,\quad r_9=w_8,\quad r_{23}=w_{24},\quad r_{24}=w_{23},
 \qquad r_j=w_j\text{ otherwise}.
\]
Let $\Pi_p$ denote this record permutation in phase $p$. The
CNOT-description frame bits and logical-sign updates are
\[
 \mathbf{x}=F^X\Pi_p\mathbf{w},\qquad \mathbf{z}=F^Z\Pi_p\mathbf{w},\qquad
 \boldsymbol{\lambda}^P=\Lambda^P\Pi_p\mathbf{w},
\]
using the matrices for the selected cycle phase.
Apply the same substitution to the syndrome-parity matrices
$A^{P,\mathrm{in/out}}$ and record constraints $C$.
The physical correction is
\[
 \mathcal{P}_p^{\rm native}=U_p^{\rm out}\mathcal{P}_p(U_p^{\rm out})^\dagger.
\]
Thus phase 1 also exchanges $X$ and $Z$ in the physical frame, whereas
phase 0 does not.

The simulations use the all-data logical representative. At $d=7$, for
either Pauli type, it is the table's logical multiplied by face checks
$4,9,11,14$. Its sign increment is consequently
\[
 \lambda_{D,p}^{P}=\lambda_p^P
 \oplus\bigoplus_{f\in\{4,9,11,14\}}
 \left[s_f^{P,\mathrm{in}}\oplus s_f^{P,\mathrm{out}}\right].
\]
For larger codes choose $L^X=L^Z=(1,\ldots,1)$. Including separate data
preparation and readout, a memory of $T=d$ cycles has $24T+3$ moments
($23T+3$ at $d=3$).

\subsection{Torus 4.8.8 circuit}
Remove the arrows from the twelve illustrated CNOT matchings to obtain
CZ edge sets $E_1,\ldots,E_{12}$. Before packing, the cycle is
\[
 R_A,\ L_0,\ \operatorname{CZ}(E_1),\ L_1,\ldots,
 \operatorname{CZ}(E_{12}),\ L_{12},\ M_A.
\]
Table~\ref{tab:native-torus-rotations} specifies the ancilla rotations in
each $L_j$, repeated in every primitive cell. Apply $H$ to every data qubit
in $L_3$ and $L_9$, and no data rotations in the other $L_j$.
The signed Pauli actions in the key fix each single-qubit Clifford up to
irrelevant global phase. Packing this sequence gives 25 moments.

\begin{table}[H]
\centering\small
\setlength{\tabcolsep}{4pt}\renewcommand{\arraystretch}{1.08}
% Generated by native_torus_audit.py after source-bound reconstruction.
\begin{tabular}{@{}c*{16}{c}@{}}
\toprule
$j$ & $0$ & $1$ & $2$ & $3$ & $4$ & $5$ & $6$ & $7$ & $8$ & $9$ & $10$ & $11$ & $12$ & $13$ & $14$ & $15$\\
\midrule
$0$ & $H$ & $H$ & $H$ & $H$ & $H$ & $H$ & $H$ & $H$ & $H$ & $H$ & $H$ & $H$ & $H$ & $H$ & $H$ & $H$\\
$1$ & $H$ & $H$ & $I$ & $I$ & $I$ & $I$ & $H$ & $H$ & $H$ & $H$ & $I$ & $I$ & $I$ & $I$ & $H$ & $H$\\
$2$ & $H$ & $I$ & $I$ & $H$ & $H$ & $I$ & $I$ & $H$ & $H$ & $I$ & $I$ & $H$ & $H$ & $I$ & $I$ & $H$\\
$3$ & $I$ & $I$ & $I$ & $I$ & $I$ & $I$ & $I$ & $I$ & $I$ & $I$ & $I$ & $I$ & $I$ & $I$ & $I$ & $I$\\
$4$ & $B$ & $I$ & $B$ & $V$ & $A$ & $H$ & $I$ & $H$ & $B$ & $I$ & $V$ & $A$ & $B$ & $I$ & $V$ & $I$\\
$5$ & $A$ & $I$ & $A$ & $B$ & $V$ & $A$ & $I$ & $B$ & $I$ & $I$ & $A$ & $V$ & $V$ & $H$ & $I$ & $I$\\
$6$ & $I$ & $A$ & $H$ & $I$ & $I$ & $H$ & $V$ & $I$ & $I$ & $H$ & $B$ & $I$ & $I$ & $B$ & $B$ & $V$\\
$7$ & $I$ & $B$ & $I$ & $I$ & $I$ & $I$ & $A$ & $I$ & $V$ & $B$ & $I$ & $I$ & $I$ & $I$ & $V$ & $I$\\
$8$ & $H$ & $I$ & $V$ & $I$ & $I$ & $V$ & $I$ & $A$ & $A$ & $I$ & $B$ & $I$ & $A$ & $I$ & $H$ & $B$\\
$9$ & $I$ & $I$ & $I$ & $I$ & $I$ & $I$ & $I$ & $I$ & $I$ & $I$ & $I$ & $I$ & $I$ & $I$ & $I$ & $I$\\
$10$ & $H$ & $I$ & $I$ & $H$ & $H$ & $I$ & $I$ & $H$ & $H$ & $I$ & $I$ & $H$ & $H$ & $I$ & $I$ & $H$\\
$11$ & $I$ & $I$ & $H$ & $H$ & $H$ & $H$ & $I$ & $I$ & $I$ & $I$ & $H$ & $H$ & $H$ & $I$ & $H$ & $I$\\
$12$ & $H$ & $H$ & $H$ & $H$ & $H$ & $H$ & $H$ & $H$ & $H$ & $H$ & $H$ & $H$ & $H$ & $H$ & $H$ & $H$\\
\bottomrule
\end{tabular}
\par\smallskip
\begin{tabular}{@{}ccll@{}}
\toprule
Symbol & Gate & $X\mapsto$ & $Z\mapsto$\\
\midrule
$I$ & identity & $X$ & $Z$\\
$H$ & Hadamard & $Z$ & $X$\\
$A$ & $C_{XYZ}$ & $Y$ & $X$\\
$B$ & $C_{ZYX}$ & $Z$ & $Y$\\
$V$ & $\sqrt{X}$ & $X$ & $-Y$\\
\bottomrule
\end{tabular}

\caption{Native torus rotations. Columns identify the sixteen ancillas by their
record indices from Table~\ref{tab:torus-frame}, and row $j$ gives the rotations in $L_j$.
The key specifies conjugation $P\mapsto UPU^\dagger$.}
\label{tab:native-torus-rotations}
\end{table}

\paragraph{Frame and syndrome conversion.}
Within the illustrated cell, define the bit vector $\mathbf{b}$ and operators $Q,U$ by
\[
 \begin{split}
 b_j&=1\text{ for }j\in\{1,2,3,7,9,13,15\},\quad b_j=0\text{ otherwise},\\
 Q&=X_{\{6,7,10,11\}}Z_{\{9,10,13,14\}},\qquad
 U=H_{\{4,5,6,7,12,13,14,15\}}.
 \end{split}
\]
Repeat these subsets under the same cell translations. With $K(\mathbf{r})$ denoting
the conditional data operation of the CNOT construction, the native
conditional operation is, up to branch-global phase,
\[
 K_{\rm native}(\mathbf{w})=UQK(\mathbf{w}\oplus\mathbf{b})U.
\]
Evaluate the matrices of Section~\ref{sec:classical-construction}
at $\mathbf{r}=\mathbf{w}\oplus\mathbf{b}$. The fixed Pauli $Q$ then gives
\[
\begin{aligned}
 s_f^{P,\mathrm{in,native}}
 &=\bigl[A^{P,\mathrm{in}}(\mathbf{w}\oplus\mathbf{b})\bigr]_f,\\
 s_f^{P,\mathrm{out,native}}
 &=\bigl[A^{P,\mathrm{out}}(\mathbf{w}\oplus\mathbf{b})\bigr]_f
   \oplus\epsilon_Q(S_f^P),\\
 \lambda_j^{P,\mathrm{native}}
 &=\bigl[\Lambda^P(\mathbf{w}\oplus\mathbf{b})\bigr]_j
   \oplus\epsilon_Q(L_j^P),
\end{aligned}
\]
where $QS Q^\dagger=(-1)^{\epsilon_Q(S)}S$.
These signs refer to the physical operators $US_f^PU^\dagger$
and $UL_j^PU^\dagger$. The ideal constraints become
$C(\mathbf{w}\oplus\mathbf{b})=\mathbf{0}$, and the physical Pauli frame is
\[
 \mathcal{P}^{\rm native}=U \mathcal{P}QU.
\]
Larger tori use the same local conversions with their own
syndrome-parity and superdense frame update matrices.

For memory initialization and termination, reset all data together with
the first ancilla reset and measure all data with the last ancilla
measurement. In zero-based cycle moments, place the initial data Hadamards
in moment 1 and the final data Hadamards in moment 23. Their support is
the support of $U$, complemented on all data for an $X$ memory.
The complete memory therefore occupies $25T$
moments, including initialization and final measurement.

\end{document}